\documentclass{article}

\usepackage{microtype}
\usepackage{graphicx}
\usepackage{subfigure}
\usepackage{booktabs}
\usepackage{hyperref}

\usepackage[accepted]{icml2023}
\usepackage{multibib}
\newcites{sec}{References (Appendix)}

\usepackage{amsmath}
\usepackage{amssymb}
\usepackage{mathtools}
\usepackage{amsthm}

\usepackage[capitalize,noabbrev]{cleveref}

\usepackage{url}
\usepackage{bm,physics}
\usepackage{graphicx,wrapfig}
\usepackage{float}
\usepackage{enumitem}
\usepackage{blkarray}
\usepackage[final]{showlabels}
\usepackage{amsfonts}       % blackboard math symbols
\usepackage{nicefrac}       % compact symbols for 1/2, etc.
\usepackage{microtype}      % microtypography
\usepackage{xcolor}         % colors

\newcommand\Ut{\tilde{\bm{U}}}
\newcommand\Wt{\tilde{\bm{W}}}
\newcommand\Sx{\bm{\Sigma_{x}}}
\newcommand\Sxi{\bm{\Sigma_{x}}^{-1}}
\newcommand\xxt{\bm{x}\bm{x}^T}

\newcommand\Gbar{\bar{\bar{\bm{\mathcal{G}}}}}
\newcommand\G{\bm{\mathcal{G}}}

\usepackage{environ}
\NewEnviron{gather+}[1][1]{%
  \begin{equation}
  \scalebox{#1}{$\begin{gathered}\BODY\end{gathered}$}
  \end{equation}
}
\usepackage{thmtools,thm-restate}
\theoremstyle{plain}
\newtheorem{theorem}{Theorem}[section]
\newtheorem{proposition}[theorem]{Proposition}

\theoremstyle{definition}

\theoremstyle{remark}
\newtheorem{remark}[theorem]{Remark}

\usepackage[textsize=tiny]{todonotes}

\icmltitlerunning{Representational Drift in a Two-Layer Neural Network}

\begin{document}

\twocolumn[

\icmltitle{Stochastic Gradient Descent-Induced Drift of Representation in a Two-Layer Neural Network}


\begin{icmlauthorlist}
\icmlauthor{Farhad Pashakhanloo}{aff1}
\icmlauthor{Alexei Koulakov}{aff1}
\end{icmlauthorlist}

\icmlaffiliation{aff1}{Cold Spring Harbor Laboratory, Cold Spring Harbor, New York, NY, USA}

\icmlcorrespondingauthor{Farhad Pashakhanloo}{pashakh@cshl.edu}

\icmlkeywords{Machine Learning, Continual Learning, Stochastic Gradient Descent, Representation Learning}

\vskip 0.3in
]

%\printAffiliations

\printAffiliationsAndNotice

\begin{abstract}
Representational drift refers to over-time changes in neural activation accompanied by a stable task performance. 
Despite being observed in the brain and in artificial networks, the mechanisms of drift and its implications are not fully understood.
Motivated by recent experimental findings of stimulus-dependent drift in the piriform cortex, we use theory and simulations to study this phenomenon in a two-layer linear feedforward network. Specifically, in a continual online learning scenario, we study the drift induced by the noise inherent in the Stochastic Gradient Descent (SGD). 
By decomposing the learning dynamics into the normal and tangent spaces of the minimum-loss manifold, we show the former corresponds to a finite variance fluctuation, while the latter could be considered as an effective diffusion process on the manifold.
We analytically compute the fluctuation and the diffusion coefficients for the stimuli representations in the hidden layer as functions of network parameters and input distribution. Further, consistent with experiments, we show that the drift rate is slower for a more frequently presented stimulus. Overall, our analysis yields a theoretical framework for better understanding of the drift phenomenon in biological and artificial neural networks.
\end{abstract}

%%%%%%%%%%%%%%%%%%%%%%%%%%%%%%%
%%%%%%%%% Introduction %%%%%%%%
%%%%%%%%%%%%%%%%%%%%%%%%%%%%%%%

\section{Introduction} \label{introduction}
Representational drift has been observed across different parts of the nervous system, such as in the hippocampus \citep{ziv2013long}, sensorimotor \citep{2019causes}, and visual systems \citep{2021visual,marks2021stimulus}.
A recent study in the piriform cortex  \citep{nature2021}, a brain structure processing information about smells, demonstrated drift of odorant representations despite stable odor identification. 
Such drift was characterized by a gradual and across-days decay in the self-similarity of the stimulus representation.
Additionally, in the same study, the drift was shown to be stimulus dependent. It was shown that a more familiar stimulus drifts at a smaller rate \citep{nature2021}.

The mechanisms and implications of the drift in the brain and artificial neural networks are still under investigation (see recent reviews by \citet{masset2022drifting}, \citet{2019causes}, and \citet{driscoll2022representational}).
Recent modeling studies have shown that drift could happen in the presence of synaptic or other types of noise \citep{qin2023coordinated,aitken2022geometry}. 
Using simulations of neural networks, \citet{aitken2022geometry} showed different noise types injected during training could lead to drift with qualitatively different patterns and geometries. However, no theoretical consideration of the drift was provided in that study. \citet{qin2023coordinated} studied the drift in a network with a similarity-matching objective and a Hebbian/anti-Hebbian learning rule. They showed that noisy synaptic updates could lead to a random-walk exploration of the solution space. Other mechanisms have been suggested on how a stable readout could be performed despite an evolving population code \citep{rule2022self, rule2020stable,kalle2021drifting}.

One important source of noise during learning for both natural and artificial neural networks is the sampling noise that arises from the stochasticity in observing the data. It is therefore natural to ask if and how this type of noise can lead to drift, and whether it can explain the stimulus-dependency of the rate of the drift observed in experiments. Here, we aim to answer these questions in a two-layer linear neural network model that undergoes online continual learning via Stochastic Gradient Descent (SGD). In addition to being common in machine learning, feed-forward networks are a reasonable first approximation to sensory (olfactory) processing in the brain. Specifically, we found the two-layer linear network to be one of the simplest models that enables studying the representational drift in the hidden layer, and yet allows for analytical tractability.

%%%%%%%%%%%%%%%%%%%%%%%%%%%%%%%
%%%%%%%Model and Theory %%%%%%%
%%%%%%%%%%%%%%%%%%%%%%%%%%%%%%%

\section{Model and Theory}
A multilayer network including $L$ layers can be described by a set of weight matrices $\bm{W}^{(l)}$, where index $l$ enumerates the individual layers.
The network can be represented by a vector in a vector space constructed from non-commuting weight matrices:
\begin{align}
\bm{\theta} = \left( \bm{W}^{(1)}, \bm{W}^{(2)}, ... , \bm{W}^{(L)} \right).
\label{vector}
\end{align}
Our goal is to study the dynamics of learning and the drift in this space for a two-layer neural network described below (Section \ref{nnmodel}). We derive the manifold of stable minimum-loss (Section \ref{degeneracy}), and study its first order differential geometry (Section \ref{diffgeomman}).
We use the Euclidean inner product in this space, which means for two vectors $\bm{\theta}_1$ and $\bm{\theta}_2$, we have:
\begin{align}
    \bm{\theta}_1^T\bm{\theta}_2 = \sum_l \tr(\bm{W}_1^{(l)T}\bar{\bm{W}}_2^{(l)}).
    \label{innerprod}
\end{align} 
By characterizing the movements normal and tangential to the manifold (Sections \ref{fluctuation} and \ref{tangupdate}), we derive an effective diffusion process on the manifold and calculate the corresponding diffusion ("drift") rates for the representations (Section \ref{diffprocess}). Finally, we apply this framework to study the dependency of the drift on the input statistics in two cases of isotropic Gaussian stimuli (Section \ref{whitenoise}), and a case with a frequently presented stimulus (Section \ref{freqstimulus}). Our results are further validated by numerical simulations.

\subsection{Neural network model} \label{nnmodel}
Our model consists of a linear feedforward neural network with an expansive hidden layer, as shown in Figure \ref{figmodel}. The input to the network is the external stimulus ($\bm{x} \in \mathbb{R}^n$), the hidden layer activation is the representation of the stimulus ($\bm{h} \in \mathbb{R}^p$, e.g. neural activities in the piriform cortex), and the output is the task outcome ($\bm{y} \in \mathbb{R}^m$, e.g. the percept or an identity of an stimulus).
\begin{figure}
\vspace{0.1in}
  \centering
\includegraphics[width=0.9\columnwidth]{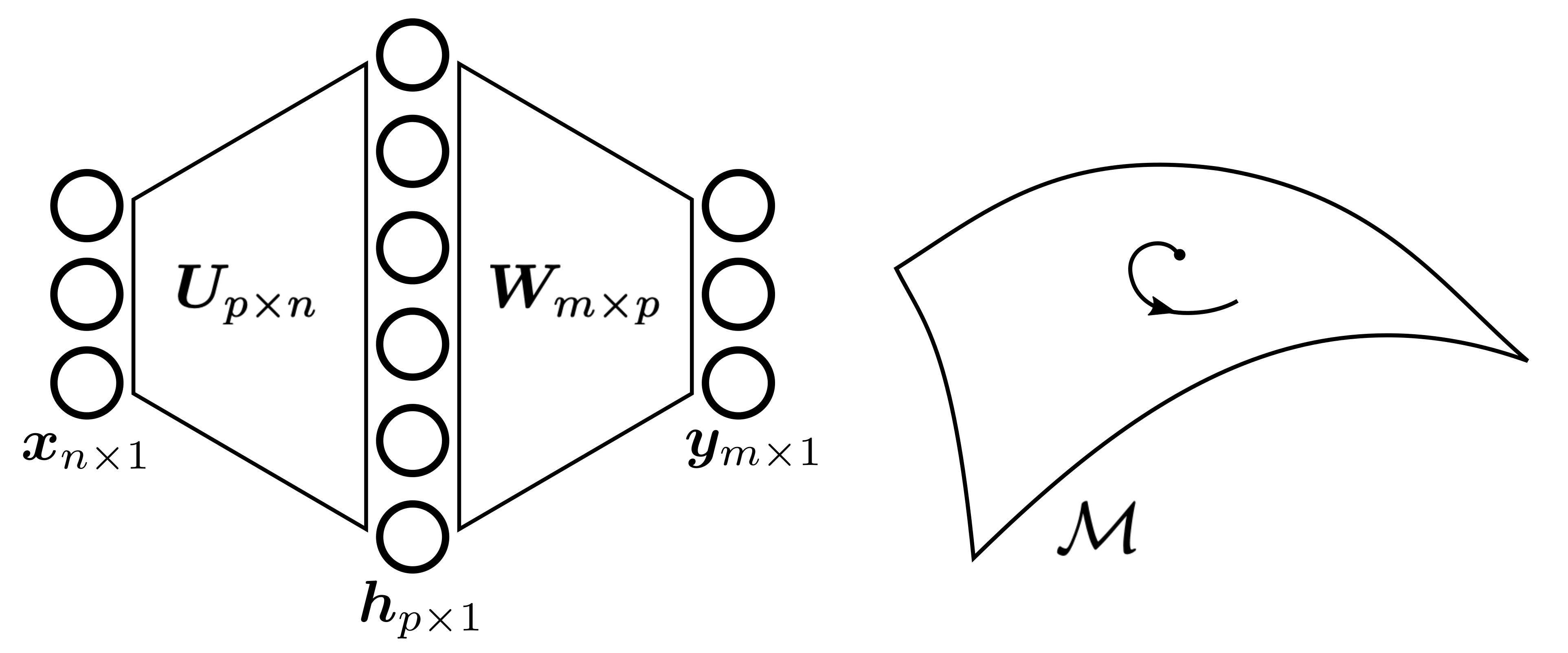}
  \caption{(left) Neural network model. (right) Manifold of minimum-loss ($\mathcal{M}$) in the parameter space.}
   \label{figmodel}
   \vspace{-0.1in}
\end{figure}
The two weight matrices are $\bm{U} = \bm{W}^{(1)}$ and $\bm{W} = \bm{W}^{(2)}$ respectively, and the predicted output is $\bm{\hat{y}} = \bm{W}\bm{U}\bm{x}$.
In the parameter space the network is denoted by $\bm{\theta} = (\bm{U},\bm{W})$.

We consider a continual online learning scenario in which the network sees one sample at a time taken independently from a fixed data distribution. The objective function consist of a Mean Squared Error (MSE) loss and L2-regularization. Hence, the sample loss becomes:
\begin{align}
 l(\bm{x},\bm{y};\bm{\theta}) =\!\frac{1}{2}\norm{\bm{y}\!-\!\bm{W}\bm{U}\bm{x}}^2 + \frac{\gamma}{2}(\norm{\bm{W}}^2_F\!+ \!\norm{\bm{U}}^2_F),
\label{loss}
\end{align}
where $\gamma$ is the regularization coefficient.  
Further, we assume $\bm{y} = \bm{x}$, which means the goal of the network is to learn the identity mapping from the input to the output. Note this essentially becomes an autoencoder with an expansive hidden layer ($p \geq n=m$). Finally, the learning occurs via the SGD with a minibatch size of one. The update in the parameter vector $\bm{\theta}$ upon observing sample $\bm{x}$ is: 
\begin{align}
    \Delta \bm{\theta} = -\eta \bm{g}(\bm{x};\bm{\theta}),
\label{sgd}
\end{align}
where $\bm{g}(\bm{x};\bm{\theta}) = \nabla_{\bm{\theta}} l(\bm{x};\bm{\theta})$ is the sample gradient vector, and $\eta$ is the learning rate.
In terms of the weight matrices, this is equivalent to
$\Delta\bm{W} = -\eta\nabla_{\bm{W}}l$ and $\Delta\bm{U} = -\eta\nabla_{\bm{U}}l$.
\subsection{Degeneracy and the manifold of solutions} \label{degeneracy}
A condition that makes representational drift possible is the redundancy of the network parameters associated with achieving a given task performance. In a two-layer network with $\bm{\hat{y}} = \bm{W}\bm{U}\bm{x}$, the transformations $\Wt \to \Wt\bm{Q}$ and $\Ut \to \bm{Q}^{-1}\Ut$ leave the output unchanged for any invertible matrix $\bm{Q}$.
In our model, the presence of L2 regularization confines $\bm{Q}$ to be orthogonal for the loss to remain unchanged. The effect of $\bm{Q}$ is rotation of the representations, and hence it links the degeneracy to a rotational symmetry in the network.
We aim to study whether and how the space of redundant parameters are explored due to the online learning stochasticity and, from that, characterize the rate and the geometry of the drift in the representations.

We define the \textit{manifold of solutions}, $\mathcal{M}$, to be the set of stable critical points in the parameter space.
This manifold represents the redundancy in the model, as all the points on it have the same expected loss value $L(\bm{\theta}) = \langle l(\bm{\theta}) \rangle_x$, and hence are equally preferable (note that $\langle . \rangle_x$ is the expectation over the input distribution).
We analytically derive $\mathcal{M}$ in the following theorem as a function of the input distribution, and for the rest of the paper refer to it as the \textit{manifold}.
\begin{restatable}{theorem}{mantheorem}
%\begin{theorem}
\label{theoremM}\textnormal{(Manifold)}
The manifold of solutions for learning the identity map ($\bm{y} = \bm{x}$) is:
\begin{gather}
\mathcal{M}\!:\!\{\tilde{\bm{\theta}}\!=\! (\Ut,\Wt)\!\mid\!\Wt\Wt^T\!=\!\bm{I}_n\!-\!\gamma \Sxi\!, \Ut = \Wt^T\},\!
\label{manifold}
\end{gather}
where $\Sx = \langle \bm{x}\bm{x}^T \rangle_x$ is the second order moment of the input distribution with eigenvalues that are assumed to be greater than $\gamma$.
%\end{theorem}
\end{restatable}
The above can be shown by first finding the critical points of the loss satisfying $\nabla_{\bm{U}}L = \nabla_{\bm{W}}L = 0$. The solutions consist of saddle-points, as well as the above manifold which has a non-negative Hessian eigenspectrum (see the proof in Appendix \ref{prooftheoremM}).
\subsection{Differential geometry of the manifold} \label{diffgeomman}
In order to study the learning dynamics near the manifold, we characterize its first order differential geometry. This is done by finding the local tangent and normal spaces to the manifold ($T_{\tilde{\bm{\theta}}}\mathcal{M}$ and $N_{\tilde{\bm{\theta}}}\mathcal{M}$, respectively).
\begin{restatable}{lemma}{tanlemma}
%\begin{lemma}
\label{lemmatangent}
\textnormal{(Tangent space)}
The local tangent space to the manifold at point $(\Wt^T,\Wt)$ 
is spanned by vectors $\bm{t} = (\bm{T_W}^T,\bm{T_W})$ with $\bm{T_W} = \Wt\Wt^T\bm{\Omega}\Wt + \bm{K}\Wt_{\bot}$, where
$\bm{\Omega}$ is an arbitrary $n \times n$ skew-symmetric matrix, $\bm{K}$ is an arbitrary $n \times (p-n)$ matrix, and $\Wt_{\perp}$ is a full-rank $(p-n) \times p$ matrix whose rows are orthogonal to rows of $\Wt$.
\end{restatable}
%\end{lemma}
The proof is presented in Appendix \ref{secTangent}. To gain some intuitions about the above results, note that movements in the tangent space create coordinated changes $\Delta\bm{W} = \bm{T_W}$ and $\Delta\bm{U} = \bm{T_W}^T$ in the weight matrices such that their effects cancel out at the output (this is because $\Delta(\Wt\Ut) \approx \Wt\bm{T_W}^T + \bm{T_W}\Ut = 0$). Additionally, these changes generate rotations of representation vectors, which is expected due to the rotational symmetry of the problem discussed in the previous section. In this context, the term in $\bm{T_W}$ containing $\bm{\Omega}$ creates rotations within the $n$-dimensional subspace of the existing representations (i.e. the column-space of $\Ut$), while the term containing $\bm{K}$ is responsible for rotations that move the representations toward the $(p-n)$-dimensional subspace orthogonal to the current representations\footnote{Note the dimension of the tangent space is $np-\frac{n(n+1)}{2}$, which is equal to the degrees of freedom in $\bm{\Omega}$ and $\bm{K}$.}.

Using the inner product in Eq.\ref{innerprod}, the normal space can be defined as the subspace orthogonal to the tangent space. We will use the eigenspace of the Hessian corresponding to positive eigenvalues to describe the normal space (see Section \ref{secstochastic} and Appendix \ref{HeigenspaceM} -- note similarly, the tangent space is equivalent to the eigenspace of the Hessian with zero eigenvalue).
In the subsequent sections, we use the results in this section to project the gradient vector into these spaces.

\subsection{Learning dynamics near the manifold and emergence of drift} \label{secstochastic}
In this section, we will study the stochastic dynamics of learning in an online learning scenario in which the data are sampled from a fixed distribution. We assume sufficient training has been performed, and hence we are on or close to the manifold while being continuously nudged around due to the SGD noise. An arbitrary point $\bm{\theta}$ near the manifold can be represented as:
\begin{align}
\bm{\theta} = \tilde{\bm{\theta}} + \bm{\theta}_N,
\end{align}
where $\tilde{\bm{\theta}}$ is the closest point on the manifold to $\bm{\theta}$, and $\bm{\theta}_N \in N_{\tilde{\bm{\theta}}}\mathcal{M}$ is the deviation from the manifold in the normal space.
In the rest of this section, we will find update equations for $\bm{\theta}_N$ and $\tilde{\bm{\theta}}$, and describe their dynamics.
\subsubsection{Fluctuation in the normal space} \label{fluctuation}
We can find an update equation for $\bm{\theta}_N$ by projecting the two sides of the SGD update (Eq.\ref{sgd}) into the normal space.
For small learning rates and over long times, this equation can be approximated with a continuous-time stochastic differential equation (SDE):
\begin{align}
d\bm{\theta}_N &= -\bm{H}\bm{\theta}_N dt + \sqrt{\eta}\,\bm{C} d\bm{B}_t.
\label{SDEthetaN}
\end{align}
In the above, $\bm{H}_{k,l \,(\in [1,2np])} = \frac{\partial L}{\partial \bm{\theta}_k\partial \bm{\theta}_l}|_{\tilde{\bm{\theta}}}$ 
is the Hessian, $\bm{B}_t$ is the standard multi-dimensional Brownian motion (Wiener process), and $\bm{C} \approx \langle\tilde{\bm{g}}\tilde{\bm{g}}^T\rangle_x^{1/2}$, where $\tilde{\bm{g}}$ is the gradient at point $\tilde{\bm{\theta}}$ on the manifold (Eq.\ref{gM}).
In general, the gradient ($\bm{g}$) can be analytically calculated at any point by differentiating the loss function with respect to the weight matrices (see Eq.\ref{eqnGradient} in the Appendix). On the manifold this becomes:
 \begin{align}
     \tilde{\bm{g}}(\bm{x}) := \bm{g}(\bm{x};\tilde{\bm{\theta}}) = (\Wt^T\bm{Z_x},\bm{Z_x}\Wt),
     \label{gM}
\end{align}
where $\bm{Z_x} = \gamma (\bm{I}_n-\Sxi\bm{x}\bm{x}^T)$. Note that $\bm{C}$
in Eq.\ref{SDEthetaN} is in general a function of $\bm{\theta}_N$. But in deriving the SDE, we approximated it with its value on the manifold ($\bm{\theta}_N = 0$), 
which is justified for small deviations (see Section \ref{derivationSDE} for the derivation of the SDE).

Since the Hessian is positive semidefinite on the manifold, the process defined by the SDE in Eq.\ref{SDEthetaN} is a mean-reverting process known as the multi-dimensional Ornstein-Uhlenbeck (OU) stochastic process (\citet{gardiner1985handbook}). OU has a stationary solution with zero mean and a finite variance (Appendix \ref{SDEsolution}). Hence, we refer to the deviations in the normal space as \textit{fluctuation}.
We represent the fluctuation in a basis constructed from the eigenvectors of the Hessian with positive eigenvalues:
\begin{gather} \label{rhodef}
\bm{\rho}   = \bm{N}^T \bm{\theta}_{N}, \quad \text{where} \,\, \bm{N} = [\bm{n}_1|\,\bm{n}_2|\,...\,|\bm{n}_K], \,\, \\\text{and}\,\, \bm{H}\bm{n}_k = \lambda_k \bm{n}_k, \,\, \lambda_k > 0. \nonumber
\end{gather}
Here $\rho_k$ represents the deviation from the manifold along the Hessian eigenvector $\bm{n}_k$.
As shown in Appendix \ref{SDEsolution}, the stationary solution of $\bm{\rho}$ satisfies $\langle \rho_k\rangle = 0$, and has the covariance:
\begin{align}
    \langle \rho_k \rho_l \rangle
    = \frac{\eta}{\lambda_k+\lambda_l} \langle {\bm{n}}_k^T\tilde{\bm{g}}(\bm{x})\,  {\bm{n}}_l^T\tilde{\bm{g}}(\bm{x}) \rangle_x.
    \label{rhocovariance}
\end{align}
As expected, since $\tilde{\bm{g}}(\bm{x})$ is the driver of the fluctuation in the SDE (Eq.\ref{SDEthetaN}), 
the fluctuation covariance depends on the projection of $\tilde{\bm{g}}$ on $\{\bm{n}_k\}$. In Appendix \ref{HeigenspaceM} we calculate these eigenvectors, which consists of two sets. We show that $\tilde{\bm{g}}$ has no projection on one of these sets and hence, for the purpose of fluctuation, only a subspace of the normal space is relevant. This subspace is described in the next proposition.
\begin{proposition} \label{propH}\textnormal{(Hessian eigenspace)}
If $(\bm{v}_i,
s_i)$ are the eigenvector/eigenvalue pairs of $\Sx =\langle \bm{x}\bm{x}^T\rangle$ with $s_1 \geqslant .. \geqslant s_n > \gamma$, the Hessian eigenvectors $\bm{n}_k$ along which there is non-zero fluctuation (i.e. $\bm{n}_k^T\tilde{\bm{g}} \neq 0$) correspond to $(\bm{v}_i,\bm{v}_j)$ pairs for $i,j \in [1,n]$ via:
\begin{align} \label{nk}
\bm{n}_{k} \equiv \bm{n}_{ij} = (\Wt^T\bm{Z}_{ij},\bm{Z}_{ij}\Wt),
\end{align}
where $\bm{Z}_{ij} = C_{ij}(\kappa_{ij}\bm{v}_i\bm{v}_j^T + \bm{v}_j\bm{v}_i^T)$. Here, $\kappa_{ij} = \text{sgn}(i-j)(\sqrt{1+b^2}-b)$
for $b = [1/\gamma - (s_i+s_j)/(2s_is_j)]|s_i-s_j|$,
and $C_{ij} =[(2-\gamma/s_i-\gamma/s_j)(1+\kappa_{ij}^2)]^{-\frac{1}{2}}$ is a normalization constant ensuring $\bm{n}_k^T\bm{n}_k = 1$.
The corresponding Hessian eigenvalues are $\lambda_{ii}=2(s_i-\gamma)$ and $\lambda_{ij(i \neq j)} = 2s_i - \gamma(s_i/s_j+ \kappa_{ij})$. (See proof in Appendix \ref{proofpropH}.)
\end{proposition} 
The above proposition relates the relevant directions in the normal space to the corresponding $(ij)$ indices of the input principal directions (this makes $k$ a composite index). Subsequently, the components of the fluctuation covariance can be written as $\langle \rho_k \rho_l \rangle \equiv \langle \rho_{ij} \rho_{pq} \rangle$ (see Appendix \ref{calcCompFluc} for the components as functions of the input distribution). The remark below provides an interpretation of movements along different $\bm{n}_k$ in the representation space. 

\begin{remark}
\label{remarkfluc} (Fluctuation of the representations):
When on the manifold, the representations of all unit-length stimuli form an n-dimensional ellipsoid embedded in $\mathbb{R}^p$. The main axes of this ellipsoid are $\tilde{\bm{h}}_i = \Wt^T\bm{v}_i$, with norms $|\tilde{\bm{h}}_i| = (1-\gamma/s_i)^{1/2}$ ($i \in [1,n]$). 
Using the above proposition,
we can determine how deviations along different $\bm{n}_k$ lead to different deformation modes for the ellipsoid.
Specifically, the deviation along $\bm{n}_{ii}$ changes $\tilde{\bm{h}}_i$ radially (changing its norm), while leaving other (orthogonal) axes unchanged.
Similarly, deviation along $\bm{n}_{ij}$ for $i \neq j$ moves $\tilde{\bm{h}}_i$ and $\tilde{\bm{h}}_j$ toward or away from each other, and leaves other axes intact.
The latter shows that the fluctuations can change not only the lengths but also the pairwise angles of the representation vectors. Note however that the average of these changes is zero over time.
\end{remark}

We take the variance of the norm as a measure of the fluctuation for the representations, i.e.
\begin{align}
    \sigma_{i}^2 := \text{var}( |\bm{h}_i(t)|) = \frac{1}{2}\langle \rho_{ii}^2\rangle = \frac{\eta\gamma^2}{4s_i}(\frac{\langle x_i^4\rangle_x}{s_i^2} - 1),
    \label{hvar}
\end{align}
where $x_{i} := \bm{v}_i^T\bm{x}$ (see Section \ref{flucrepnorm} for derivation).
As we show later, this has an excellent match to numerically measured fluctuations.

\subsubsection{Tangential updates} \label{tangupdate}
Over a learning update, the displacement in the parameter space ($\Delta\bm{\theta}$) leads to a corresponding projected movement on the manifold ($\Delta \tilde{\bm{\theta}}$). For small learning rates and by avoiding the curvature effect, we have: 
\begin{align}
\Delta \tilde{\bm{\theta}} = -\eta \bm{g}_{\scriptscriptstyle T}(\bm{x};\tilde{\bm{\theta}}\!+\! \bm{N}\!\bm{\rho}),
\label{sgdT}
\end{align}
where $\bm{g}_{\scriptscriptstyle T} = \Pi_T(\bm{g})$ is the projection of the gradient vector onto the tangent space.
In the next theorem, we calculate this projection and find its action on the representations.
\begin{restatable}{theorem}{tantheorem}
%\begin{theorem}
\label{theoremphi}
\textnormal{(Tangential update)} For a normal deviation $\bm{N}\!\bm{\rho}\!=\!\sum_k \rho_k \bm{n}_k$ from point $\tilde{\bm{\theta}} = (\Wt^T,\Wt)$ on the manifold, the tangential projection of the gradient is:
\begin{gather}
    \bm{g}_{\scriptscriptstyle T}(\bm{x};\tilde{\bm{\theta}}\!+\! \bm{N}\!\bm{\rho}) = ( \bm{G}_{{\bm{U}}_T}, \bm{G}_{{\bm{U}}_T}^T), \quad\text{for}\label{gT} \\ \bm{G}_{{\bm{U}}_T}\!=\!\Wt^T (\Wt\Wt^T)^{-\frac{1}{2}} (\sum_k{\rho_k}\G_{k}^{:,:})(\Wt\Wt^T)^{\frac{1}{2}}\!+\!\mathcal{O}(|\bm{\rho}|^2), \nonumber
\end{gather}
where $\G$ is a rank-3 tensor with components
 $\G_{k}^{s,r} \equiv \G_{ij}^{s,r}$ that are defined as:
\begin{multline}
     \G_{ij}^{s,r}(\bm{x}) =\frac{\gamma C_{ij}\sqrt{\omega_r\omega_s}}{2(\omega_r + \omega_s)} [x_j(S^{ij}_{sj}x_s \delta_i^r - S^{ij}_{rj}x_r \delta_i^s)
    \\+ x_i(S^{ij}_{is}x_s \delta_j^r - S^{ij}_{ir}x_r \delta_j^s)],
\label{G}
\end{multline}
where $\delta_s^r$ is the Kronecker delta function, $\omega_i := 1 - \gamma/s_i$, and $S_{rs}^{ij} := \kappa_{ij}/s_r + 1/s_s$ ($r,s,i,j \in [1,n]$).
The action of the tangential gradient in the representation space is a small rigid-body rotation around the origin. In this rotation, the representation vectors stay within the column-space of $\Wt^T$, which itself stays fixed during the tangential updates. Additionally, the angular displacement of the representation $\bm{h}_s (= \bm{U}\bm{v}_s)$ toward $\bm{h}_r (= \bm{U}\bm{v}_r)$ is:
\begin{align}
    \Delta {\varphi}_{sr}(\bm{x};\tilde{\bm{\theta}}\!+\! \bm{N}\!\bm{\rho})
= \eta\sum_{k=1}^{K}{\rho_k}\G_{k}^{s,r}(\bm{x}) + \mathcal{O}(|\bm{\rho}|^2).
\label{dphi}
\end{align}
%\end{theorem}
\end{restatable}
In deriving the above, the gradient was approximated to the first order of the deviation from the manifold, and subsequently the tangential projector operator was applied. Here $\G$ is a tensor that, for a given direction in the normal space (corresponding to $k \equiv ij$), returns a rotation generator via the skew-symmetric matrix $\G_{k}^{:,:}$ ($G^{s,r}_k = -G^{r,s}_k$). Note also that during the gradient updates only a subspace of the tangent space is explored corresponding to the term with $\bm{\Omega}$ in Lemma \ref{lemmatangent} (see Appendix \ref{prooftheoremphi}).
\subsubsection{Drift as a diffusion process} \label{diffprocess}
As we saw in Section \ref{fluctuation}, the mean-reverting property of the gradient confines the normal deviations near the manifold (Fig.\ref{figschem} - Left). The movements in the tangent space, however, face no such mean-reverting property and hence can diffuse freely on the manifold in a random-walk fashion. As per Theorem \ref{theoremphi}, the tangential component of the gradient is proportional to the deviation from the manifold ($\bm{\rho}$), and it vanishes on the manifold (see the schematics in the middle and right panels of Fig.\ref{figschem}). 
We can find an effective diffusion process for the movements on the manifold, which because of the above correspondence, is expected to depend on the statistics of the normal deviations.
Specifically,
over large timescales,
Eq.\ref{sgdT} can be approximated as the following continuous-time SDE, describing the evolution of the the parameter vector on the manifold:
\begin{align}
d \tilde{\bm{\theta}} =  \sqrt{2\eta^{-1}\bm{D}_{\theta}}\,d\bm{B}_t.
\label{SDEtheta}
\end{align}
In the above, $\bm{D}_{\theta} \!=\!(1/2)\langle \Delta \tilde{\bm{\theta}} \Delta \tilde{\bm{\theta}}^T\rangle_{x,\rho}\!=(\eta^2/2)\langle \bm{g}_{\scriptscriptstyle T} \bm{g}_{\scriptscriptstyle T}^T\rangle_{x,\rho}$
is the diffusion tensor for the parameters (measured over a training step), and $\bm{B}_t$ is multi-dimensional Brownian motion.
\begin{figure}
\vskip 0.1in
\centering
\includegraphics[width=\columnwidth]{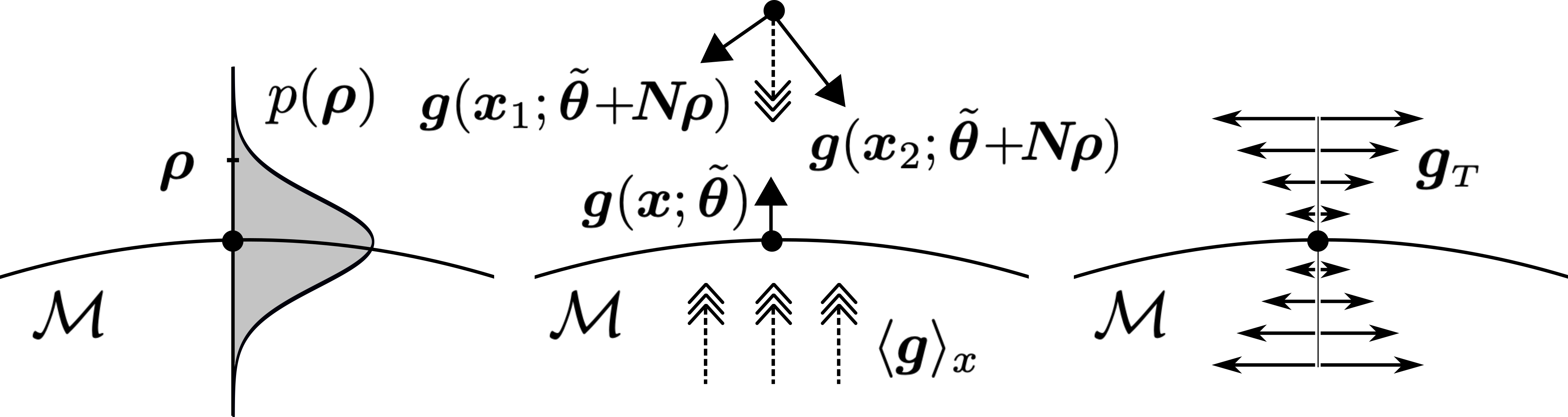}
\vskip -0.05in
\caption{Schematics showing (left) the probability distribution of the fluctuations outside the manifold, (middle) the gradient vectors on and near the manifold, and (right) the tangential component of the gradient.}
\label{figschem}
\vskip -0.15in
\end{figure}
The above diffusion process is manifested as a rotational diffusion in the representation space. According to Theorem \ref{theoremphi}, the rotations happen within the $n$-dimensional space of representations and hence to fully characterize the diffusion it's sufficient to measure the pairwise diffusion rates between the axes of the ellipsoid $\{{\bm{h}}_s\}$ (${s \in [1,n]}$). Similar to the parameter diffusion above, we define ${D}_{sr}$ to be the diffusion rate between the two representations ${\bm{h}}_s$ and ${\bm{h}}_r$ based on the mean squared of the pairwise angular displacement, i.e.:
\begin{align} \label{Dsr}
    {D}_{sr} &:= \frac{1}{2} \langle \Delta{\varphi}_{sr}^2 \rangle_{x,\rho}   =\frac{\eta^2}{2} \sum_{k,l=1}^{K} \langle \rho_k \rho_l \rangle \,\Gbar_{k,l}^{s,r},
\end{align}
where we defined $\Gbar_{k,l}^{s,r} := \langle \G_{k}^{s,r}\G_{l}^{s,r} \rangle_x$.
The right-hand side results from replacing $\Delta{\varphi}_{sr}$ from Eq.\ref{dphi} and taking the average.
The total diffusion for the representation $\bm{h}_s$ can be derived by summing up the diffusion rates along different directions: $D_{s} = \sum_{r=1, r\neq s}^{n} D_{sr}$.

Eq.\ref{Dsr} indicates that the diffusion rate between two representation vectors is an aggregation of terms resulting from deviations in different directions in the normal space. In this equation, $\eta^2\Gbar_{k,l}$ can be thought of as the diffusion per direction (or more accurately direction pair $k,l$) in the normal space, and $\langle \rho_k \rho_l \rangle$ is the covariance of the corresponding fluctuations. As both $\Gbar$ and $\langle \bm{\rho}\bm{\rho}^T \rangle$ depend on the second and the fourth moments of the stimuli distribution, the diffusion is also a function of the those moments of the input distribution.

\subsection{Numerical simulations}
Alongside the analytical derivations, we also performed numerical simulations in which we measured the drift in a neural network undergoing continual SGD training. We first numerically validated the equations of the manifold by verifying that after enough training steps the network approaches the theoretical manifold. Next, to measure the drift, we initialized multiple ($>10^4$) realizations of the network, all starting from a fixed point on the manifold but undergoing different random SGD samplings during the training. Following a transitory phase, we studied the over-time trajectories of the hidden layer activation for different trial stimuli. This was done by measuring the fluctuation (variance) of the representation norm, and the angular displacements of the representations. The diffusion coefficient was estimated as half of a linear fit slope to the mean squared angular displacements aggregated from all the realizations.

%%%%%%%%%%%%%%%%%%%%%%%%%%%%%%%%%%%
%% Drift under Gaussian stimuli %%%
%%%%%%%%%%%%%%%%%%%%%%%%%%%%%%%%%%%
\section{Drift Under Gaussian Stimuli} \label{whitenoise}
In this section, we present complete analytical results for a case where the stimuli are drawn randomly and independently from a standard n-dimensional Gaussian distribution, i.e. $x_i \sim \mathcal{N}(0,1)$. Since the input covariance is the identity ($\langle x_ix_j\rangle = \delta_i^j$), its eigenvectors $\{\bm{v}_i\}$ form an orthonormal basis for $\mathbb{R}^n$, all corresponding to  eigenvalue $s_i = 1$.
Using Proposition \ref{propH}, the relevant directions in the normal space consist of three sets corresponding to $\bm{Z}_{ii}\!=\!  \bm{v}_i\bm{v}_i^T/(\sqrt{2\omega})$, $\bm{Z}_{ij(i>j)}\!=\!(\bm{v}_i\bm{v}_j^T + \bm{v}_j\bm{v}_i^T)/(2\sqrt{\omega})$, and $\bm{Z}_{ij(i<j)}\!=\!(-\bm{v}_i\bm{v}_j^T + \bm{v}_j\bm{v}_i^T)/(2\sqrt{\omega})$, where $\omega=1-\gamma$, and $i,j \in [1,n]$. For calculating the above, we derived the coefficients $\kappa_{ij(i>j)}=1$, $\kappa_{ii} = 0$, and $\kappa_{ij(i<j)}=-1$ in the proposition.
The corresponding Hessian eigenvalues are $\lambda_{ij(i \geqslant j)} = 2(1-\gamma)$ and $\lambda_{ij(i < j)} = 2-\gamma$.

The components of the fluctuation matrix result from projecting the gradient onto the above eigenvectors (Eq.\ref{rhocovariance}) and using the identity $\langle x_ix_jx_px_q\rangle = \delta_{ip}^{jq} + \delta_{ij}^{pq} + \delta_{ij}^{qp}$ for the Gaussian input (see Appendix \ref{calcCompFluc}). This leads to:
\begin{gather}
\langle \rho_{ij}^2\rangle = \eta\gamma^2, \, \text{for}\,\,i\geqslant j,
\label{rhoID}
\end{gather}
with the rest of the components being zero. Following Eq.\ref{hvar}, the fluctuation in the representation norm of an arbitrary unit-length stimulus becomes:
\begin{align}
    \sigma_{s}^2  = \frac{\eta\gamma^2}{2}.
\end{align}
To find the diffusion coefficients, we first calculate the components of the $\G$ tensor from Eq.\ref{G}:
\begin{gather} \label{GtermsID}
\G_{ii}^{s,r} = \frac{\gamma}{2\sqrt{2\omega}}x_i(x_s\delta_i^r - x_r\delta_i^s),\quad \G_{ij(i<j)}^{s,r} = 0,\\ 
\G_{ij(i>j)}^{s,r} = \frac{\gamma}{4\sqrt{\omega}}[x_j(x_s\delta_i^r - x_r\delta_i^s) + x_i(x_s\delta_j^r - x_r\delta_j^s )]. \nonumber
\end{gather}
Subsequently, the coefficients $\Gbar^{s,r}_{ij,pq}$ can be found by taking the averages of the products of the $\G$ terms from Eq.\ref{GtermsID}, leading to:
\begin{align}
    &\Gbar^{s,r}_{ii,ii} =\frac{\gamma^2}{8\omega}(\delta_i^r + \delta_i^s),\, \Gbar^{s,r}_{ii,jj(i\neq j)} =\frac{-\gamma^2}{8\omega}(\delta_{ij}^{rs} + \delta_{ij}^{sr}), \nonumber \\  
    &\Gbar^{s,r}_{ij,ij(i>j)} \!=\!\frac{\gamma^2}{16\omega}(\delta_i^r\!+\!\delta_i^s\!+\! \delta_j^r\!+\!\delta_j^s\!+\! 2\delta_{ij}^{rs}+2\delta_{ij}^{sr}),\!
\label{GbarID}
\end{align}
with the rest of the components being zero. Finally, replacing the above in the summation of Eq.\ref{Dsr} results in the diffusion rate between the representations of two arbitrary stimuli denoted by $s$ and $r (\neq s)$:
\begin{align}
    D_{sr} &= \frac{\eta^3\gamma^2}{2}( \sum_{i=1}^{n}\Gbar^{s,r}_{ii,ii}  + \sum_{i,j=1, i>j}^{n}\Gbar^{s,r}_{ij,ij}) \nonumber\\ &= \frac{\eta^3 \gamma^4}{2(1-\gamma)}(\frac{1}{4} + \frac{n}{8}) = \frac{1}{16}\frac{\eta^3 \gamma^4}{1-\gamma}(n+2).
\end{align}
Note the summation was performed only over indices where the fluctuation covariance was non-zero.
Since the diffusion is isotropic, the total diffusion for representation $\tilde{\bm{h}}_s$ becomes:
\begin{align}
    &D_{s} = (n-1)D_{sr} = \frac{1}{16}\frac{\eta^3 \gamma^4}{1-\gamma} (n-1)(n+2).
\end{align}
%%%%%%%%%%%%%%%%%%%%%%%%%%%%%%%%%%%
%% Drift under a freq. stimulus %%%
%%%%%%%%%%%%%%%%%%%%%%%%%%%%%%%%%%%
\section{Drift Under a Frequent Stimulus} \label{freqstimulus}
As demonstrated in the previous section,  SGD-induced drift could occur even with isotropic background stimuli. In this section, we will study how the presence of a frequent stimulus in the environment influences the drift.
To do this, we consider a case where in addition to the background Gaussian stimuli, there is a relatively more frequent stimulus, $\bm{a}$, that is presented with probability $\alpha$ during training, i.e.:
\begin{align}\label{xfreq}
    \bm{x} = \, \left\{ \begin{array}{lc}
\bm{a} & \text{Pr} = \alpha \\
\mathcal{N}(\bm{0},\bm{I}_n) & \text{Pr} = 1-\alpha
\end{array} \right.
\end{align}
Note that the previous case in Section \ref{whitenoise} is equivalent to $\alpha=0$, and here we study the drift as a function of $\alpha$. Without loss of generality, we take $\bm{a}$ to be along the first axis, and for simplicity we assume it has unit length. The second and fourth moments of the input distribution become $\langle x_ix_j\rangle = \alpha \delta_{ij}^{11} + (1-\alpha)\delta_i^j$ and $\langle x_ix_jx_px_q \rangle = \alpha \delta_{ijpq}^{1111} + (1-\alpha)(\delta_{ip}^{jq} + \delta_{ij}^{pq} + \delta_{ij}^{qp})$, respectively. The eigenvectors of $\Sx$ still form an orthonormal basis for the input, but with eigenvalues $s_a = s_1 = 1$, and $s_b = 1-\alpha$ for $b \in [2,n]$ (note we use indices $a$ and $b$ to refer to the frequent and orthogonal background stimuli, respectively).
The fluctuations of the representation norm for the frequent stimulus $\bm{a}$ and a unit length background stimulus $\bm{b} \, (\perp\!\bm{a})$ can be found from Eq.\ref{hvar}:
\begin{align}
    &\sigma_{a}^2  
    = \frac{\eta\gamma^2}{2}(1-\alpha), \quad
    \sigma_{b}^2     
    =\frac{\eta\gamma^2}{4}\frac{2+\alpha}{(1-\alpha)^2}. 
\end{align}
It is easy to check that $\sigma_{a}^2 \leqslant \sigma_{b}^2$ irrespective of $\alpha$, which suggests that the fluctuation of the frequent stimulus representation is smaller than that of a background stimulus.

To fully characterize the diffusion in the representation space, we take advantage of the symmetry within the space of background stimuli. This means we only need to find two pairwise diffusion coefficients: $D_{ab}$ $(= D_{ba})$, which measures the diffusion between the frequent and a background stimulus, and $D_{bc}$ $(= D_{cb})$, which measures the pairwise diffusion between two orthogonal background stimuli
(without loss of generality, we take $b=2$ and $c=3$). To calculate these coefficients, we need to perform the summation in Eq.\ref{Dsr} which is taken over indices $(k,l)$ or equivalently $(ij,pq)$ for $i,j,p,q \in [1,n]$. For large $n$, many of the terms in the summation can be ignored as shown in Appendix \ref{approxsum}. Here we present results for two limiting cases.

We first study the regime of small $\alpha$, which can be considered as a perturbation to the previous Gaussian stimuli case.
As a first order correction to the summation, it is sufficient to sum over indices for which the unperturbed case had non-zero $\langle \rho_{ij}\rho_{pq} \rangle$ or $\Gbar^{s,r}$ terms. This limits the summation to indices  $(ij,ij)_{i>j}$ for which either $i$ or $j$ are equal to $r$ or $s$ (see details in Appendix \ref{approxsum}). If we take $d$ $(>c)$ to be an additional index within the background subspace (e.g. $d=4$), the diffusion summations for the two coefficients simplify to:
\begin{align} \label{Dsrsmallalpha}
D_{ba} &\approx \frac{n\eta^2}{2}(\langle\rho_{db}^2\rangle\Gbar_{db,db}^{a,b} + \langle \rho_{da}^2\rangle \Gbar_{da,da}^{a,b}) \quad (\alpha \ll 1, n\gg 1)
\nonumber \\
&\approx \frac{n\eta^3\gamma^4}{32}[(1+2\alpha)(1) + (1+\frac{\alpha}{2})(1+\alpha)] \nonumber \\ 
&\approx \frac{n\eta^3\gamma^4}{16} (1 + \frac{7\alpha}{4})\nonumber \\
D_{bc} &\approx \frac{n\eta^2}{2}(\langle \rho_{db}^2 \rangle \Gbar_{db,db}^{c,b} + \langle \rho_{dc}^2 \rangle \Gbar_{dc,dc}^{c,b}) \nonumber \\
&\approx \frac{n\eta^3\gamma^4}{32}[(1+2\alpha)(1+\alpha) + (1+2\alpha)(1+\alpha)] \nonumber \\
&\approx \frac{n\eta^3\gamma^4}{16} (1 + 3\alpha),
\end{align} 
where we replaced the following quantities calculated in Appendix \ref{Gcalcfreqstim} under small $\alpha$ and $\gamma$:
\begin{align}
&\langle \rho_{db}^2 \rangle\!=\!\langle \rho_{dc}^2 \rangle = \eta\gamma^2(1+2\alpha), \, \langle \rho_{da}^2\rangle\!=\!\eta\gamma^2(1 + \frac{\alpha}{2}) \\ 
&\Gbar_{db,db}^{a,b} = \frac{\gamma^2}{16}, \quad \Gbar_{da,da}^{a,b} = 
\Gbar_{db,db}^{c,b} = \Gbar_{dc,dc}^{c,b} = \frac{\gamma^2}{16}(1+\alpha). \nonumber
\end{align}
The above shows $D_{bc} - D_{ba} \approx 5n\eta^3\gamma^4\alpha/64 \geq 0$.
This difference can be attributed to both the fluctuation and the $\Gbar$ terms (for example, one can check from Eq.\ref{Dsrsmallalpha} that if the fluctuation terms are replaced with their values at $\alpha=0$, the difference still remains but to a lesser degree).

Total diffusion for stimuli $a$ and $b$ becomes:
\begin{align}
&D_a = (n-1)D_{ab} \approx \frac{n^2\eta^3\gamma^4}{16} (1 + \frac{7\alpha}{4})\quad (\alpha \ll 1, n\gg 1) \nonumber \\
&D_b = D_{ba} + (n-2)D_{bc} \approx \frac{n^2\eta^3\gamma^4}{16} (1 + 3\alpha), 
\end{align}
which shows $D_b \geq D_a$. 

In Appendix \ref{CalcfreqstimLargeN}, we also perform derivations for large $n$ and $\alpha \gg \gamma$ (specifically, with respect to the input eigenvalues, we assume $s_a,s_b,s_a-s_b \gg \gamma$). The results of those calculations are:
\begin{align} \label{Ds_largealpha}
D_{a} &\approx \frac{n^2\eta^3\gamma^4\langle x_a^2x_c^2\rangle\langle x_b^2x_c^2\rangle}{128s_b^5}[1 + 3\frac{s_b}{s_a} + 2(\frac{s_b}{s_a})^2  + \frac{4(\frac{s_b}{s_a})^2}{1+\frac{s_b}{s_a}}] \nonumber \\
&= \frac{n^2\eta^3\gamma^4}{16(1-\alpha)^3} [\frac{1-\frac{7}{4}\alpha+\frac{15}{16}\alpha^2-\frac{1}{8}\alpha^3}{1-\frac{\alpha}{2}}], \nonumber
\\
\nonumber\\
D_{b} &\approx \frac{n^2\eta^3\gamma^4\langle x_b^2x_c^2 \rangle^2}{16s_b^5}
= \frac{n^2\eta^3\gamma^4}{16(1-\alpha)^3} \,\,\quad (\alpha \gg \gamma, n \gg 1)
\end{align}
where in the equalities we replaced $s_a=1$, $s_b=1-\alpha$, and $\langle x_a^2x_c^2\rangle = \langle x_b^2x_c^2\rangle = 1-\alpha$. The term inside the second line brackets in $D_a$ is always smaller than one. Hence, we again have $D_b \geq D_a$.

In Figure \ref{figfreq}, we plot the diffusion and fluctuations coefficients as a function of $\alpha$ for stimuli $a$ and $b$.
We see that, consistent with the above results, the frequent stimulus drifts at a relatively slower rate and has a smaller fluctuation, irrespective of $\alpha$. Additionally, there is an excellent match between the theoretical and the simulations results for both the fluctuations and the diffusion coefficients. 
Finally, in the bottom panel of the same figure, we plotted the trajectories of the representations over time for $n=p=3$. The lower fluctuation and diffusion rates for the more frequent stimulus can be visually observed from the smaller point cloud for this stimulus.
\begin{figure}%[H]
\centering
\includegraphics[width=\columnwidth]{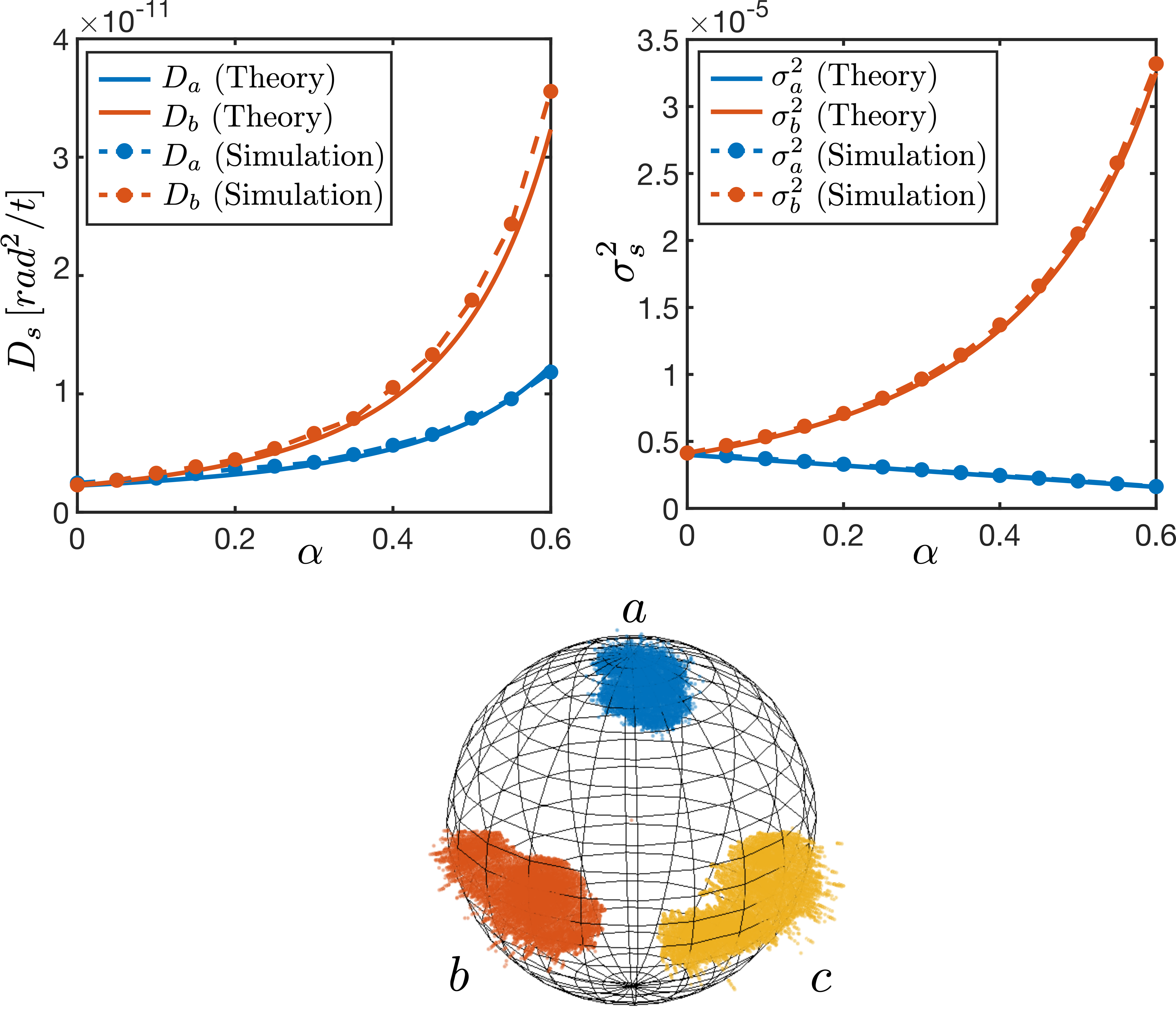}
\caption{Plots of (left) diffusion, and (right) fluctuation for representations of a frequent (subscript $a$, blue) and a background stimulus (subscript $b$, red). Horizontal axes are the probability of the frequent stimulus, $\alpha$. In both plots $m=n=10,p=20$, $\gamma=0.04$, and $\eta=0.005$. (bottom) History of representations for three trial stimuli after $2.2 \times 10^{5}$ training steps. $n=p=3$, $\gamma=0.1$, $\eta=0.1$, and $\alpha=0.5$. $a$ is the frequent and $b$ and $c$ are two orthogonal background stimuli.}
 \label{figfreq}
\end{figure}

%%%%%%%%%%%%%%%%%%%%%%%%%%%%%%%%%%%
%%%%%%%%%%% Discussion %%%%%%%%%%%%
%%%%%%%%%%%%%%%%%%%%%%%%%%%%%%%%%%%
\section{Discussion}
In a two-layer neural network model of the olfactory system, we show, using theory and simulations, that the stochasticity in SGD online learning could result in a drift of stimuli representation over time, even after the training is complete and no measurable change in the performance is observed. We analytically demonstrate the dependency of the drift on the input distribution and, in particular, show that a frequently presented stimulus drifts at a relatively slower rate. This finding is consistent with experimental observations in the piriform cortex \citep{nature2021}.

We studied the learning dynamics in the high-dimensional space of network parameters. In this space, drift can be considered as any movement tangential to the manifold of minimum-loss, the effects of which aggregate over time to create an effective diffusion process. Orthogonal to this is the fluctuation outside the manifold that has a finite variance due to the mean-reverting property of the gradient. For the tangential gradient to have a non-zero value, an orthogonal deviation from the manifold was necessary (see Eq.\ref{gT} and Figure \ref{figschem}). In a way, this makes the diffusion on the manifold a second order phenomenon, and that explains why the amount of diffusion depends on the fluctuation covariance (Eq.\ref{Dsr}). In our work, diffusion (random walk) emerged as a mechanism underlying the representational drift. Note that in this sense it is different from the term drift in Physics.

In the representation space, the effects of the fluctuation and tangential movements are deformations and rigid-body rotations of the space, respectively. If we consider the representation of a given stimulus over time, its trajectory consists of two parts: a random walk movement on a high-dimensional sphere, and simultaneously, a mean-reverting fluctuating process that causes deviations on and outside the sphere (i.e. changing norms and pair-wise angles of representations). This can be observed from the visualization in the bottom panel of Figure \ref{figfreq}. Note that the random walk on a sphere dynamics is consistent with the over-time decay of population response self-similarity as observed in the experiments of \citep{nature2021}.
Further, the lower diffusion rate observed for a more frequent stimulus suggests that the rigid-body rotation on average rotates the frequent stimulus to a lesser degree. This could be explained by, for example, the axis of rotation being on average closer to the representation of the frequent stimulus. The lower diffusion and fluctuation for the more frequent stimulus could be traced to the mechanism by which the gradient tends to preserve the representation and the output of more rehearsed stimuli..

Recently, \citet{qin2023coordinated} studied drift in a one-layer neural network model of sensory systems with a similarity matching objective and a Hebbian/anti-Hebbian learning rule. Similarly to our results, they showed that the representations undergo a rotational diffusion process. This is not surprising as both models have objective functions with degenerate solutions and rotational symmetry. In our case, the L2 regularization creates that symmetry. 
\citet{qin2023coordinated} also demonstrated a stimulus dependency of the diffusion coefficient, such that the drift for a given eigenvector direction was inversely related to its eigenvalue.
However, in that study synaptic noise was injected to the network and the stimulus dependency was not studied when this noise was zero. Here, we used a two-layer neural network model with MSE loss to study drift under no external noise. We showed that the sampling noise due to the SGD stochasticity is enough to demonstrate the stimulus dependency of the drift rate. Along that line, we speculate that the anisotropic profile of the SGD noise enhances the stimulus dependency of the drift when compared to an isotropic (synaptic) noise.
Our results further show that in the pure SGD case, the drift direction is limited to the subspace of existing stimuli representations. When an isotropic noise is injected to the weights, this will no longer be the case since the whole tangent space will be explored.

SGD dynamics have been studied extensively in the machine learning literature (see e.g. \citet{mandt2017stochastic, jastrzkebski2017three, zhu2018anisotropic,chaudhari2018stochastic,yaida2018fluctuation,zhangentropy}). In a recent study, \citet{li2022what} provided a mathematical framework for investigating the limiting dynamics of SGD around the local minimizer manifold in overparametrized networks. They used this framework to study the regularization effect of SGD (i.e. implicit bias), and showed that label noise could drive the network to areas with flatter landscape with potential benefits for generalization \cite{wei2019noise, blanc2020implicit,cowsik2022flatter}. Unlike that study, our analysis deals with a regime where the system has reached a final stage of learning with a steady local loss landscape. Hence in our case the effective dynamics on the manifold is a pure random-walk. Additionally, the analytical tractability of our model allowed us to calculate the exact gradient profile around the manifold. We found that to accurately calculate the drift, we had to integrate the effect of the tangential gradient not just on the manifold but away from it (hence the third power of the learning rate in our diffusion rate, one of which stems from the normal fluctuation covariance). Finally, our model contains weight decay, which influences both the learned manifold and the gradients.

Linear multilayer networks, despite their simplicity, demonstrate nonlinear and nontrivial learning behaviors \cite{saxe2013exact, li2021statistical}.
In our case, the two-layer feedforward network was the simplest architecture that provided necessary redundancy in the network parameters through the product of the weight matrices.
We used the identity mapping (autoencoder) setup to simplify the derivations of main formalisms, allowing us to focus on the role of input distribution on the drift. The dependence of the drift on the data is inherent in the diffusion term in Eq.\ref{Dsr}, which ultimately depends on the 2\textsuperscript{nd}/4\textsuperscript{th} moments of the input distribution.
Extensions of our work may include nonlinear and deep networks with recurrent connections in addition to a non-stationary data-distribution. The main concepts on which our framework builds are valid for more general linear or nonlinear networks minimizing an objective function with stationary data. In those cases, the loss could still be approximated as a quadratic function near the manifold of minimum loss, and hence the equations for fluctuations and tangent movements still hold. However, the specific form of the gradient as the driver of the drift and the geometry of the manifold will be problem-dependent.
We expect our main findings (i.e. the presence of drift, and lower drift for a more frequent stimulus) to carry over to nonlinear cases, however, we leave that to future studies.

\section*{Acknowledgements}
We are grateful for fundings from the Swartz Foundation, and the National Institute of Health (U19NS112953-01).

\bibliography{drift}

\begin{thebibliography}{28}
\providecommand{\natexlab}[1]{#1}
\providecommand{\url}[1]{\texttt{#1}}
\expandafter\ifx\csname urlstyle\endcsname\relax
  \providecommand{\doi}[1]{doi: #1}\else
  \providecommand{\doi}{doi: \begingroup \urlstyle{rm}\Url}\fi

\bibitem[Absil et~al.(2009)Absil, Mahony, and Sepulchre]{absil2009optimization}
Absil, P.-A., Mahony, R., and Sepulchre, R.
\newblock Optimization algorithms on matrix manifolds.
\newblock In \emph{Optimization Algorithms on Matrix Manifolds}. Princeton
  University Press, 2009.

\bibitem[Aitken et~al.(2022)Aitken, Garrett, Olsen, and
  Mihalas]{aitken2022geometry}
Aitken, K., Garrett, M., Olsen, S., and Mihalas, S.
\newblock The geometry of representational drift in natural and artificial
  neural networks.
\newblock \emph{PLOS Computational Biology}, 18\penalty0 (11):\penalty0
  e1010716, 2022.

\bibitem[Blanc et~al.(2020)Blanc, Gupta, Valiant, and
  Valiant]{blanc2020implicit}
Blanc, G., Gupta, N., Valiant, G., and Valiant, P.
\newblock Implicit regularization for deep neural networks driven by an
  ornstein-uhlenbeck like process.
\newblock In \emph{Conference on learning theory}, pp.\  483--513. PMLR, 2020.

\bibitem[Chaudhari \& Soatto(2018)Chaudhari and
  Soatto]{chaudhari2018stochastic}
Chaudhari, P. and Soatto, S.
\newblock Stochastic gradient descent performs variational inference, converges
  to limit cycles for deep networks.
\newblock In \emph{2018 Information Theory and Applications Workshop (ITA)},
  pp.\  1--10. IEEE, 2018.

\bibitem[Cowsik et~al.(2022)Cowsik, Can, and Glorioso]{cowsik2022flatter}
Cowsik, A., Can, T., and Glorioso, P.
\newblock Flatter, faster: scaling momentum for optimal speedup of sgd.
\newblock \emph{arXiv preprint arXiv:2210.16400}, 2022.

\bibitem[Deitch et~al.(2021)Deitch, Rubin, and Ziv]{2021visual}
Deitch, D., Rubin, A., and Ziv, Y.
\newblock Representational drift in the mouse visual cortex.
\newblock \emph{Current Biology}, 31\penalty0 (19):\penalty0 4327--4339, 2021.

\bibitem[Driscoll et~al.(2022)Driscoll, Duncker, and
  Harvey]{driscoll2022representational}
Driscoll, L.~N., Duncker, L., and Harvey, C.~D.
\newblock Representational drift: Emerging theories for continual learning and
  experimental future directions.
\newblock \emph{Current Opinion in Neurobiology}, 76:\penalty0 102609, 2022.

\bibitem[Edelman et~al.(1998)Edelman, Arias, and Smith]{edelman1998geometry}
Edelman, A., Arias, T.~A., and Smith, S.~T.
\newblock The geometry of algorithms with orthogonality constraints.
\newblock \emph{SIAM journal on Matrix Analysis and Applications}, 20\penalty0
  (2):\penalty0 303--353, 1998.

\bibitem[Gardiner et~al.(1985)]{gardiner1985handbook}
Gardiner, C.~W. et~al.
\newblock \emph{Handbook of stochastic methods}, volume~3.
\newblock springer Berlin, 1985.

\bibitem[Jastrzkebski et~al.(2017)Jastrzkebski, Kenton, Arpit, Ballas, Fischer,
  Bengio, and Storkey]{jastrzkebski2017three}
Jastrzkebski, S., Kenton, Z., Arpit, D., Ballas, N., Fischer, A., Bengio, Y.,
  and Storkey, A.
\newblock Three factors influencing minima in sgd.
\newblock \emph{arXiv preprint arXiv:1711.04623}, 2017.

\bibitem[Kalle~Kossio et~al.(2021)Kalle~Kossio, Goedeke, Klos, and
  Memmesheimer]{kalle2021drifting}
Kalle~Kossio, Y.~F., Goedeke, S., Klos, C., and Memmesheimer, R.-M.
\newblock Drifting assemblies for persistent memory: Neuron transitions and
  unsupervised compensation.
\newblock \emph{Proceedings of the National Academy of Sciences}, 118\penalty0
  (46):\penalty0 e2023832118, 2021.

\bibitem[Kunin et~al.(2019)Kunin, Bloom, Goeva, and Seed]{pmlr-v97-kunin19a}
Kunin, D., Bloom, J., Goeva, A., and Seed, C.
\newblock Loss landscapes of regularized linear autoencoders.
\newblock In Chaudhuri, K. and Salakhutdinov, R. (eds.), \emph{Proceedings of
  the 36th International Conference on Machine Learning}, volume~97 of
  \emph{Proceedings of Machine Learning Research}, pp.\  3560--3569. PMLR,
  09--15 Jun 2019.
\newblock URL \url{https://proceedings.mlr.press/v97/kunin19a.html}.

\bibitem[Li \& Sompolinsky(2021)Li and Sompolinsky]{li2021statistical}
Li, Q. and Sompolinsky, H.
\newblock Statistical mechanics of deep linear neural networks: The
  backpropagating kernel renormalization.
\newblock \emph{Physical Review X}, 11\penalty0 (3):\penalty0 031059, 2021.

\bibitem[Li et~al.(2022)Li, Wang, and Arora]{li2022what}
Li, Z., Wang, T., and Arora, S.
\newblock What happens after {SGD} reaches zero loss? --a mathematical
  framework.
\newblock In \emph{International Conference on Learning Representations}, 2022.
\newblock URL \url{https://openreview.net/forum?id=siCt4xZn5Ve}.

\bibitem[Mandt et~al.(2017)Mandt, Hoffman, and Blei]{mandt2017stochastic}
Mandt, S., Hoffman, M.~D., and Blei, D.~M.
\newblock Stochastic gradient descent as approximate bayesian inference.
\newblock \emph{arXiv preprint arXiv:1704.04289}, 2017.

\bibitem[Marks \& Goard(2021)Marks and Goard]{marks2021stimulus}
Marks, T.~D. and Goard, M.~J.
\newblock Stimulus-dependent representational drift in primary visual cortex.
\newblock \emph{Nature communications}, 12\penalty0 (1):\penalty0 1--16, 2021.

\bibitem[Masset et~al.(2022)Masset, Qin, and Zavatone-Veth]{masset2022drifting}
Masset, P., Qin, S., and Zavatone-Veth, J.~A.
\newblock Drifting neuronal representations: Bug or feature?
\newblock \emph{Biological cybernetics}, pp.\  1--14, 2022.

\bibitem[Qin et~al.(2023)Qin, Farashahi, Lipshutz, Sengupta, Chklovskii, and
  Pehlevan]{qin2023coordinated}
Qin, S., Farashahi, S., Lipshutz, D., Sengupta, A.~M., Chklovskii, D.~B., and
  Pehlevan, C.
\newblock Coordinated drift of receptive fields in hebbian/anti-hebbian network
  models during noisy representation learning.
\newblock \emph{Nature Neuroscience}, pp.\  1--11, 2023.

\bibitem[Rule \& O’Leary(2022)Rule and O’Leary]{rule2022self}
Rule, M.~E. and O’Leary, T.
\newblock Self-healing codes: How stable neural populations can track
  continually reconfiguring neural representations.
\newblock \emph{Proceedings of the National Academy of Sciences}, 119\penalty0
  (7):\penalty0 e2106692119, 2022.

\bibitem[Rule et~al.(2019)Rule, O’Leary, and Harvey]{2019causes}
Rule, M.~E., O’Leary, T., and Harvey, C.~D.
\newblock Causes and consequences of representational drift.
\newblock \emph{Current opinion in neurobiology}, 58:\penalty0 141--147, 2019.

\bibitem[Rule et~al.(2020)Rule, Loback, Raman, Driscoll, Harvey, and
  O'Leary]{rule2020stable}
Rule, M.~E., Loback, A.~R., Raman, D.~V., Driscoll, L.~N., Harvey, C.~D., and
  O'Leary, T.
\newblock Stable task information from an unstable neural population.
\newblock \emph{Elife}, 9:\penalty0 e51121, 2020.

\bibitem[Saxe et~al.(2013)Saxe, McClelland, and Ganguli]{saxe2013exact}
Saxe, A.~M., McClelland, J.~L., and Ganguli, S.
\newblock Exact solutions to the nonlinear dynamics of learning in deep linear
  neural networks.
\newblock \emph{arXiv preprint arXiv:1312.6120}, 2013.

\bibitem[Schoonover et~al.(2021)Schoonover, Ohashi, Axel, and Fink]{nature2021}
Schoonover, C.~E., Ohashi, S.~N., Axel, R., and Fink, A.~J.
\newblock Representational drift in primary olfactory cortex.
\newblock \emph{Nature}, pp.\  1--6, 2021.

\bibitem[Wei \& Schwab(2019)Wei and Schwab]{wei2019noise}
Wei, M. and Schwab, D.~J.
\newblock How noise affects the hessian spectrum in overparameterized neural
  networks.
\newblock \emph{arXiv preprint arXiv:1910.00195}, 2019.

\bibitem[Yaida(2018)]{yaida2018fluctuation}
Yaida, S.
\newblock Fluctuation-dissipation relations for stochastic gradient descent.
\newblock \emph{arXiv preprint arXiv:1810.00004}, 2018.

\bibitem[Zhang et~al.(2018)Zhang, Saxe, Advani, and Lee]{zhangentropy}
Zhang, Y., Saxe, A.~M., Advani, M.~S., and Lee, A.~A.
\newblock Energy–entropy competition and the effectiveness of stochastic
  gradient descent in machine learning.
\newblock \emph{Molecular Physics}, 116\penalty0 (21-22):\penalty0 3214--3223,
  2018.
\newblock \doi{10.1080/00268976.2018.1483535}.
\newblock URL \url{https://doi.org/10.1080/00268976.2018.1483535}.

\bibitem[Zhu et~al.(2018)Zhu, Wu, Yu, Wu, and Ma]{zhu2018anisotropic}
Zhu, Z., Wu, J., Yu, B., Wu, L., and Ma, J.
\newblock The anisotropic noise in stochastic gradient descent: Its behavior of
  escaping from sharp minima and regularization effects.
\newblock \emph{arXiv preprint arXiv:1803.00195}, 2018.

\bibitem[Ziv et~al.(2013)Ziv, Burns, Cocker, Hamel, Ghosh, Kitch, El~Gamal, and
  Schnitzer]{ziv2013long}
Ziv, Y., Burns, L.~D., Cocker, E.~D., Hamel, E.~O., Ghosh, K.~K., Kitch, L.~J.,
  El~Gamal, A., and Schnitzer, M.~J.
\newblock Long-term dynamics of ca1 hippocampal place codes.
\newblock \emph{Nature neuroscience}, 16\penalty0 (3):\penalty0 264--266, 2013.

\end{thebibliography}
\bibliographystyle{icml2023}

%%%%%%%%%%%%%%%%%%%%%%%%%%%%%%%%%%%%%%%%%%%%%%%%%%%%%%%%%%%%%%%%%%%%%%%%%%%%%%%
%%%%%%%%%%%%%%%%%%%%%%%%%%%%%%%%%%%%%%%%%%%%%%%%%%%%%%%%%%%%%%%%%%%%%%%%%%%%%%%
% APPENDIX
%%%%%%%%%%%%%%%%%%%%%%%%%%%%%%%%%%%%%%%%%%%%%%%%%%%%%%%%%%%%%%%%%%%%%%%%%%%%%%%
%%%%%%%%%%%%%%%%%%%%%%%%%%%%%%%%%%%%%%%%%%%%%%%%%%%%%%%%%%%%%%%%%%%%%%%%%%%%%%%
\newpage
\appendix
\nocitesec{*}
\onecolumn
%\begin{center}
{\large \textbf{Appendix}}
%\end{center}
%\section{Appendix}

\section{Derivation of the Manifold and the Hessian (Proof of Theorem \ref{theoremM})} \label{prooftheoremM}
%\mantheorem*
We lay out different parts of the proof over the next few sections. We first calculate all the critical points that satisfy the zero expected gradient condition (Section \ref{criticalpoints}), and from those, we calculate the Hessian (sections \ref{Heigenspace} and \ref{HeigenspaceM}), and show that the manifold $\mathcal{M}$ is the only stable solution with non-negative Hessian spectrum (Section \ref{instability}). The proof is summarized in 
Section \ref{proofMsummary}.

\subsection{Derivation of the critical points} \label{criticalpoints}
The gradient can be found by taking the derivative of the loss function in Eq.\ref{loss}. By assuming $\bm{y}=\bm{x}$, we have:
\begin{align}
\bm{g}(\bm{x};\bm{\theta}) = \nabla_{\bm{\theta}} l(\bm{x};\bm{\theta}) = (\bm{G_U},\bm{G_W}), \,\,
\left\{ \begin{array}{l}
 \bm{G_W}= \nabla_W l = (\bm{W}\bm{U}-\bm{I}_n)\bm{x}\bm{x}^T\bm{U}^T  +\gamma\bm{W} \\
 \bm{G_U}=\nabla_U l = \bm{W}^T(\bm{WU}-\bm{I}_n)\bm{x}\bm{x}^T + \gamma\bm{U}
 \end{array} \right.
 \label{eqnGradient}
\end{align}
At the critical points the expected gradients are zero. This is equivalent to:
\begin{align}
\langle \bm{g}(\bm{x};\bm{\theta})\rangle_x = \nabla_{\bm{\theta}}L  = 0 \,\Longrightarrow \,
\left\{ \begin{array}{c}
 (\bm{I}_n -\bm{WU})\Sx\bm{U}^T = \gamma\bm{W}  \\
 \bm{W}^T(\bm{I}_n-\bm{WU})\Sx = \gamma\bm{U}
  \end{array} \right.
  \label{WU}
\end{align}
which need to be solved for $\bm{W} \in \mathbb{R}^{n \times p}$ and $\bm{U} \in \mathbb{R}^{p \times n}$ under $p \geq n$. A similar set of equations was studied in \citet{pmlr-v97-kunin19a} but with $p \leq n$. We follow part of a proof from that study (Proof 2.1 in that paper) to show two facts about the solutions of the above equations (note that despite different assumptions for the dimension of the hidden layer, the following facts still hold).

\underline{\textit{\text{Claim 1}}}: $\bm{C} := (\bm{I}_n - \bm{W}\bm{U})\Sx \succeq 0$.
\textit{Proof}.
By multiplying the first line of Eq.\ref{WU} by $\bm{W}^T$ from the right we get $(\bm{I}_n - \bm{W}\bm{U})\Sx(\bm{W}\bm{U})^T = \gamma\bm{W}\bm{W}^T \succeq 0$, and therefore $\Sx(\bm{W}\bm{U})^T \succeq (\bm{W}\bm{U})\Sx(\bm{W}\bm{U})^T$. Following a property of positive semi-definite matrices that states if $A \succeq 0$ and $AB^T \succeq BAB^T$, then $A \succeq BA$, we have $\Sx \succeq \bm{W}\bm{U}\Sx$, which proves the claim. $\qedsymbol{}$ 

\underline{\textit{\text{Claim 2}}}: $\bm{U}=\bm{W}^T$. \textit{Proof}.
By subtracting the transpose of the first line of Eq.\ref{WU} from the second, and using the symmetric property of $\bm{C}$, we have: $(\bm{U}-\bm{W}^T)(\bm{C}+\gamma\bm{I}_n)=0$. Now, since $\bm{C}+\gamma\bm{I}_n \succ 0$, we use a property of positive definite matrices that states if $B \succ 0$ and $A^T BA=0$, then $A=0$, to get $\bm{U}-\bm{W}^T = 0$. $\qedsymbol{}$ For more information see \citet{pmlr-v97-kunin19a} and its supplementary material.

Replacing $\bm{W}$ from the first line of Eq.\ref{WU} into its second line, and using $\bm{U}=\bm{W}^T$ from Claim 2, we can form an equation for $\bm{W}\bm{W}^T$:
\begin{align}
(\bm{I}_n -\bm{W}\bm{W}^T)\Sx (\bm{W}\bm{W})^T (\bm{I}_n -\bm{W}\bm{W}^T)\Sx = \gamma^2 \bm{W}\bm{W}^T
\label{WWt}
\end{align}
To solve the above we note that because we showed that $\bm{C} = (\bm{I}_n -\bm{W}\bm{W}^T)\Sx$ is semi-positive definite and hence symmetric, two matrices $\bm{W}\bm{W}^T$ and $\Sx$ commute and hence can be diagonalized simultaneously. Therefore, if $\Sx = \bm{V}\bm{S_V}\bm{V}^T$ is the singular value decomposition (SVD) with $\bm{S_V} = \text{diag}(\{s_i\})$, we can assume $\bm{W}\bm{W}^T = \bm{V}\bm{\Lambda}\bm{V}^T$, where $\bm{\Lambda} = \text{diag}(\{\Lambda_i\})$ for $i \in [1,n]$.
Replacing these in Eq.\ref{WWt}, leads to $n$ separate equations for the diagonal entries $\Lambda_i$:
\begin{align}
\Lambda_i(\Lambda_i^2s_i^2 -2\Lambda_is_i^2 + (s_i^2-\gamma^2)) = 0, \quad i \in [1,n]
\end{align}
A valid solution for each entry is either $\Lambda_i=0$ or $\Lambda_i = 1 - \frac{\gamma}{s_i}$ (note $\Lambda_i = 1 + \frac{\gamma}{s_i}$ also satisfies the above equation but it's not a valid solution overall since any $\bm{\Lambda}$ with such an entry violates the positive semi-definiteness of $\bm{C}$). For a given $\bm{\Lambda}$, we have $\bm{W} = \bm{V}\bm{\Lambda}^{\frac{1}{2}}\bm{Q}^T$ and $\bm{U} = \bm{W}^T = \bm{Q}\bm{\Lambda}^{\frac{1}{2}T}\bm{V}^T$ for any orthonormal matrix $\bm{Q}$.
Based on the values of diagonal entries, we can classify the critical points into three sets:
\begin{itemize}
\item \underline{\textit{The manifold ($\mathcal{M}$):}} $\Lambda_i = 1 - \frac{\gamma}{s_i}$ for all $i \in [1,n]$. In this case $\bm{\Lambda} = \bm{I}_n - \gamma\bm{S_V}^{-1}$ and hence $\bm{W}\bm{W}^T = \bm{I}_n - \gamma\Sxi$.
\item \underline{\textit{Zero solution:}} $\Lambda_i = 0$ for all $i \in [1,n]$. In this case $\bm{W}$ and $\bm{U}$ are both zero.
\item \underline{\textit{Mixed solutions:}} $\Lambda_i = 0$ for some $i \in [1,n]$, and $\Lambda_i = 1 - \frac{\gamma}{s_i}$ for the other entries.
\end{itemize}

In the next section we show, by calculating the Hessian spectrum, that the manifold ($\mathcal{M}$) is a stable solution, and the zero and the mixed solutions are unstable critical points. We will also use the results of the Hessian in other parts of the paper.

\subsection{Hessian eigenspace} \label{Heigenspace}
To find the eigenvectors of the Hessian, we assume we are a small deviation away from a critical point $\tilde{\bm{\theta}} = (\Ut,\Wt)$ along the vector $\bm{n}$.  The eigenvalue equation can be written as:
\begin{align}
\bm{n} =  (\bm{N_U},\bm{N_W}): \qquad \bm{H}\bm{n} = \lambda \bm{n} \,\equiv \, \left\{ \begin{array}{c}
\nabla_{W}L\,|_{{\tilde{\bm{\theta}} + \bm{n}}} = \lambda\bm{N_W}
 \\
 \nabla_{U}L \,|_{{\tilde{\bm{\theta}} + \bm{n}}} = \lambda\bm{N_U}
\end{array} \right.
 \label{hesseqns}
\end{align}
After replacing $\bm{n}$ in the gradient equations (Eq.\ref{eqnGradient}), and ignoring second and higher order terms, we have the following matrix equations for $\bm{N_W}$, $\bm{N_U}$ and $\lambda$:
\begin{align}
  \left\{ \begin{array}{c}
 (\Wt\Ut-\bm{I}_n)\Sx \bm{N_U}^T + (\bm{N_W}\Ut + \Wt\bm{N_U})\Sx\Ut^T + \gamma\bm{N_W} = \lambda\bm{N_W} \\
 \bm{N_W}^T(\Wt\Ut-\bm{I}_n)\Sx  + \Wt^T(\bm{N_W}\Ut + \Wt\bm{N_U})\Sx + \gamma\bm{N_U}  = \lambda\bm{N_U}  \end{array} \right.
  \label{hessmateqns}
\end{align}
To find the eigenvectors and the eigenvalues, we won't need to calculate the Hessian matrix explicitly. However, for completeness we mention here. Vectorizing the above equations lead to:
\begin{gather}
\resizebox{0.92\textwidth}{!}
{
$
[(\bm{\tilde{U}}\Sx \Ut^T + \gamma\bm{I}_p) \otimes \bm{I}_n]\text{vec}(\bm{N_W}) + [(\bm{I}_p \otimes (\Wt\Ut-\bm{I}_n)\Sx)\bm{K}^{(p,n)} + (\Ut\Sx \otimes \Wt)]\text{vec}(\bm{N_U}) = \lambda \text{vec}(\bm{N_W})
$
}\\
\resizebox{\textwidth}{!}
{$
[(\Sx(\Wt\Ut-\bm{I}_n)^T \otimes \bm{I}_p)\bm{K}^{(n,p)} + (\Sx\Ut^T \otimes \Wt^T)]\text{vec}(\bm{N_W}) + [(\Sx \otimes \Wt^T\Wt) + \gamma(\bm{I}_n \otimes \bm{I}_p)]\text{vec}(\bm{N_U}) = \lambda \text{vec}(\bm{N_U})$
} \nonumber
\end{gather}
where the $"\text{vec}"$ operator reshapes a matrix to a vector, and $\otimes$ denotes the Kronecker product between two matrices. If we consider the parameter space as $\bm{\theta} \equiv \text{vec}(\bm{U})\oplus\text{vec}(\bm{W})$, the $\scriptstyle 2np \times 2np$ Hessian matrix can be formed as below:
\begin{equation}
\bm{H} = 
\left( \begin{array}{c|c}
\bm{H_{UU}} & \bm{H_{UW}} \\
\hline
\bm{H_{WU}} & \bm{H_{WW}} \end{array} \right)
\end{equation}
\begin{align*}
&\bm{H_{WW}} = (\bm{\tilde{U}}\Sx \Ut^T + \gamma\bm{I}_p) \otimes \bm{I}_n\\
&\bm{H_{UU}} = (\Sx \otimes \Wt^T\Wt) + \gamma(\bm{I}_n \otimes \bm{I}_p) \\
&\bm{H_{WU}} = \bm{H_{UW}^T} = (\bm{I}_p \otimes (\Wt\Ut-\bm{I}_n)\Sx)\bm{K}^{(p,n)} + (\Ut\Sx \otimes \Wt)
\end{align*} 
where $\bm{K}$ is the Commutation matrix.

\subsubsection{Derivation of the Hessian eigenspace for the manifold} \label{HeigenspaceM}
Here we find the Hessian eigenspace for the manifold of Theorem \ref{theoremM}. Replacing the equations of the manifold ($\Wt\Ut = \bm{I} - \gamma\Sxi$ and $\Ut = \Wt^T$) in Eq.\ref{hessmateqns} results in the following equations for $\bm{N_W}$, $\bm{N_U}$ and $\lambda$:
\begin{align}
  \left\{ \begin{array}{c}
-\gamma \bm{N_U}^T + (\bm{N_W}\Wt^T + \Wt\bm{N_U})\Sx\Wt + (\gamma-\lambda)\bm{N_W} = 0 \\
 -\gamma \bm{N_W}^T + \Wt^T(\bm{N_W}\Wt^T + \Wt\bm{N_U})\Sx + (\gamma-\lambda) \bm{N_U}  = 0\end{array} \right.
 \label{hesseqnsM}
\end{align}
Since the Hessian is a real and symmetric matrix, the set of its eigenvectors are orthogonal and span $\mathbb{R}^{2np}$. A subspace of this eigenspace corresponds to $\lambda=0$, which is the tangent space of the manifold. We will study the tangent space in more detail in the next section, but here we note that its dimension is $np - \frac{n(n+1)}{2}$, which results from the degrees of freedom (DOF) in $\Wt$ in the equation of the manifold ($\Wt\Wt^T = \text{fixed}$). In this section, we will find the eigenspace corresponding to $\lambda \neq 0$ (i.e. the normal space), which, as a result, has the dimension $np + \frac{n(n+1)}{2}$.

We propose two sets of solutions for the equations above for $\lambda \neq 0$. The first one satisfies:
\begin{align} \label{Hsol1}
\text{Set 1:}\quad \bm{N_U} = -\bm{N_W}^T, \quad \bm{N_W} = \Wt\Wt^T\bm{S}\Wt + \bm{K}\Wt_{\perp}, \quad \text{with}\,\,\, \lambda = 2\gamma.
\end{align}
where $\bm{S}$ is an arbitrary $n \times n$ symmetric matrix, $\bm{K}$ is an arbitrary $n \times (p-n)$ matrix, and $\Wt_{\perp}$ is a full-rank $(p-n) \times p$ matrix whose rows are orthogonal to rows of $\Wt$.
(it is straightforward to check that the above satisfies Eq.\ref{hesseqnsM} as the term in the parentheses vanish i.e. $\bm{N_W}\Wt^T - \Wt\bm{N_W}^T=0$). The total DOF in this solution is $np - \frac{n(n-1)}{2}$ resulting from DOF in $\bm{S}$ and $\bm{K}$. The second solution is:
\begin{align} \label{Hsol2}
\text{Set 2:}\quad \bm{N_W} = \bm{Z}\Wt, \quad \bm{N_U} = \Wt^T\bm{Z},\quad \text{where}\,\bm{Z}\, \text{satisfies} \,\,
 -\gamma\bm{Z}^T + 2\bm{Z}\Sx - \gamma\Sxi\bm{Z}\Sx -\lambda\bm{Z} = 0.
\end{align}

To solve the equation for $\bm{Z} \in \mathbb{R}^{n \times n}$ and $\lambda$, let's assume $(\bm{v}_i,s_i)$ are the eigenvectors/eigenvalues of $\Sx$ arranged in descending order, i.e. $s_1\geq..\geq s_n$. By replacement, we can show that $n^2$ solutions exists for $(\bm{Z},\lambda)$ such that each solution corresponds to a pair of input eigenvectors $(\bm{v}_i,\bm{v}_j)$. Specifically, the solutions consists of:
\begin{align} \label{Hsol2i}
& \left \{ \begin{array}{lll}
\bm{Z}_{ii} = C_{ii}\bm{v}_i\bm{v}_i^T, & \lambda_{ii} =2(s_i-\gamma) & i \in [1,n]\\ [8pt] 
\bm{Z}_{ij} = C_{ij}(\kappa_{ij}\bm{v}_i\bm{v}_j^T + \bm{v}_j\bm{v}_i^T), & \lambda_{ij} = 2s_i - \gamma(\frac{s_i}{s_j}+ \kappa_{ij}) & i,j \in [1,n], i\neq j
\end{array} \right. \\
 &\quad \,\,\,\, \text{for}\,\, \kappa_{ij} = \text{sgn}(i-j)(\sqrt{1+b^2}-b),\,\, \text{where}\,\, b = (\frac{1}{\gamma} - \frac{s_i+s_j}{2s_is_j})|s_i-s_j|. \nonumber
\end{align}
The normalization constants $C_{ij}$ can be found by imposing $\bm{n}_{k}^T\bm{n}_k = 1$, which leads to $C_{ij} = [(1+\kappa_{ij}^2)(\omega_i + \omega_j)]^{-\frac{1}{2}}$, where $\omega_i := 1- \frac{\gamma}{s_i}$. It is also easy to check that $\bm{n}_k^T\bm{n}_l = 0$ for $k \neq l$.
Since the number of solutions in this set is $n^2$, we see that the total number of solutions in the two sets for the Hessian eigenspace ($1$ and $2$) add up to the dimension of the normal space mentioned above. This shows that these two sets are the complete solutions.

\underline{\textit{\text{Claim}}}: For $i < j$, $\lambda_{ij}>\lambda_{ji}>0$. (\textit{Proof}. The left inequality results directly from writing the eigenvalues as a function of $s_i$ and $s_j$ noting that $s_i\geq s_j > \gamma$. For the second inequality, first note that $0<\kappa_{ji} \leq 1$. Additionally, we have $s_j(2 - \frac{\gamma}{s_i}) > \gamma$, which proves $\lambda_{ji} = 2s_j - \gamma(\frac{s_j}{s_i} + \kappa_{ji}) > 0$ \qedsymbol{}).
\textit{As a result, all the Hessian eigenvalues for the manifold of solution are non-negative, which makes the manifold $\mathcal{M}$ a stable solution.}

\subsubsection{Instability of the other critical points} \label{instability}
Here we show that the alternate critical points found in Section \ref{criticalpoints} are unstable by showing that their Hessian spectrum contains negative eigenvalues.

\underline{Zero solution}: In this case all $\Lambda_i$ are zero and we have $\Wt=0$ and $\Ut=0$. By replacement in Eq.\ref{hessmateqns} we get:
\begin{align}
  \left\{ \begin{array}{c}
 \Sx \bm{N_U}^T  = (\gamma-\lambda)\bm{N_W} \\
 \bm{N_W}^T\Sx  = (\gamma-\lambda)\bm{N_U}  \end{array} \right.
\end{align}
Since $\Sx$ is full-rank, to have non-zero solutions to the above we should have $\lambda \neq \gamma$. Replacing $\bm{N_U}$ from the second line in the first one leads to the equation $(\Sx^2 -  (\gamma-\lambda)^2\bm{I}_n)\bm{N_W}=0$. For a non-zero solution to exist, the determinant of the term in the parentheses should be zero. This is equivalent to the following characteristic equations: $s_i^2 - (\gamma-\lambda)^2 = 0$, for $i \in [1,n]$. Recalling from the assumptions that $s_i>\gamma$, there is a positive and a negative $\lambda$ associated to each $s_i$:
 $\lambda_{+} = \gamma + s_i$, and $\lambda_{-} = \gamma - s_i$. \textit{Therefore, the zero solution is an unstable saddle point}.

\underline{Mixed solutions}:
Recall from Section \ref{criticalpoints} that for the mixed solution $\bm{W} = \bm{V}\bm{\Lambda}^{\frac{1}{2}}\bm{Q}^T$, where $\bm{\Lambda}$ is a diagonal matrix with some entries equal to zero. For a given mixed solution, let $\mathcal{I} \subset [1,n]$ denote the indices of the zero entries of $\bm{\Lambda}$, and $\mathcal{J} = \mathcal{I}^{\perp}$ the indices of the non-zero diagonal entries. We can show that for the mixed solution, a part of the Hessian eigenspace (corresponding to indices $\mathcal{J}$) are similar to that of the manifold described in Section \ref{HeigenspaceM}. For example, it is easy to check that for any $j \in \mathcal{J}$ and $\bm{Z}_{jj} = C_{jj}\bm{v}_j\bm{v}_j^T$,  the solution $\bm{N_W} = \bm{Z}_{jj}\Wt$, $\bm{N_U} = \Wt^T\bm{Z}_{jj}$ satisfy the Hessian equations (Eq.\ref{hessmateqns}) with positive eigenvalues $\lambda_{jj} = 2(s_j-\gamma)$. Similar to the case of the manifold, this part of the eigenspace corresponds to positive eigenvalues. However, in the mixed case, other solutions to the Hessian equation exist that have negative eigenvalues. Specifically, if $\bm{q}_i$ are columns of $\bm{Q}$, for any $i \in \mathcal{I}$ and  $\bm{\alpha} \in \text{span}(\{\bm{q}_r\})$ for $r \in [1,p]-\mathcal{J}$, the solution $\bm{N_W} = \bm{N_U}^T = \bm{v}_i\bm{\alpha}^T$ satisfy Eq.\ref{hessmateqns} with the eigenvalue $\lambda = \gamma - s_i$, which is negative. Hence, in this case we have both negative and positive eigenvalues. \textit{This proves that the mixed solutions are unstable saddle points.}

\subsection{Summary of the proof of Theorem \ref{theoremM}} \label{proofMsummary}
In Section \ref{criticalpoints} we found all the critical points and classified them into three sets consisting of the manifold ($\mathcal{M}$), the zero, and the mixed solutions. In Section \ref{HeigenspaceM}, we showed that the Hessian eigenvalues are non-negative for the manifold (see the Claim at the end of the section). Finally, in Section \ref{instability}, we showed that the other critical points are unstable saddle points. This proves that $\mathcal{M}$ consists of all the stable critical points in the parameter space. \qedsymbol{}

\newpage
\section{Tangent Space to the Manifold} \label{secTangent}
In this section, we describe the first order differential geometry of the manifold ($\mathcal{M}$) which will  allow us to study the dynamics near the manifold. In the two equations describing the manifold in Theorem \ref{theoremM}, $\Ut$ is simply the transpose of $\Wt$, and hence the manifold can be studied by considering matrix $\Wt \in \mathbb{R}^{n \times p}$ in the equation $\Wt\Wt^T = \bm{I}_n - \gamma\Sxi$. This defines a matrix manifold in $\mathbb{R}^{n \times p}$. In the case of $\gamma=0$, this simplifies to the Steifel manifold for which there are known results (see \citet{edelman1998geometry} and \citet{absil2009optimization}). Inspired by those results, here we consider a more general but similar case of $\gamma \neq 0$. The next lemma describes the tangent space to the manifold.
\tanlemma*
\begin{proof}
First, we show that any vector $\bm{t}$ of the form above belongs to the tangent space of the manifold. This is equivalent to showing that if we are on the manifold and get displaced by vector $\varepsilon\, \bm{t} = (\varepsilon \bm{T_W}^T, \varepsilon \bm{T_W})$ for a small $\varepsilon$, we stay on the manifold to the first order of $\varepsilon$ (i.e. the equations that define the manifold still hold to that order). From the two equations, the transpose equation ($\Ut = \Wt^T$) is automatically satisfied as we have $\bm{T_U}=\bm{T_W}^T$. For the other equation ($\Wt\Wt^T=\bm{I}-\gamma\Sxi$), we show that the change in $\bm{W}\bm{W}^T$ is of the second order:
\begin{align}
    \Delta(\bm{W}\bm{W}^T) &=
    \varepsilon \Wt \bm{T_W}^T + \varepsilon\bm{T_W} \Wt^T + \mathcal{O}(\varepsilon^2)\\
    &= \varepsilon\Wt\Wt^T\bm{\Omega}^T\Wt\Wt^T + \varepsilon\Wt\Wt_{\perp}^T\bm{K}^T + \varepsilon\Wt\Wt^T\bm{\Omega}\Wt\Wt^T +  \varepsilon\bm{K}\Wt_{\perp}\Wt^T + \mathcal{O}(\varepsilon^2)  \nonumber
    \\ &= \mathcal{O}(\varepsilon^2) \nonumber
    %\label{deltaWWt}
\end{align}
where in going from the second to the third line we used $\bm{\Omega}=-\bm{\Omega}^T$ and $\Wt_{\perp}\Wt^T=0$.

Next, we demonstrate that the dimension of the vector space spanned by $\bm{t}$ is the same as the dimension of the tangent space. The dimension of the tangent space can be inferred from the difference between the parameters and the constraints in $\Wt\Wt^T = \bm{I}-\gamma\Sxi$, which are $np$ and $n(n+1)/2$ respectively. Hence, the dimension of the tangent space is $np - n(n+1)/2$.
To find the dimension of the space spanned by $\bm{t}$, note that the transformation from $(\bm{\Omega},\bm{K})$ to  $\bm{T_W}$ is one-to-one. This is simply because each of the mappings $\bm{\Omega} \to \Wt\Wt^T\bm{\Omega}\Wt$ and $\bm{K} \to \bm{K}\Wt_{\perp}$ are one-to-one (as $\Wt$ and $\Wt_{\perp}$ are full rank), and their ranges do not overlap (as the rows of $\Wt$ and $\Wt_{\perp}$ are orthogonal).
As a result, the dimension of the space spanned by $\bm{T_W}$ is equal to the dimension of the space spanned by the skew-symmetric $\bm{\Omega}$, which is $n(n-1)/2$, added to that of $\bm{K}$, which is $n(p-n)$. This becomes the same as the dimension of tangent space mentioned above. 
Since we already showed that any vector $\bm{t}$ lies in the tangent space, the set of $\bm{t}$'s span the tangent space. This completes the proof.
\end{proof}

With the definition of the Euclidean inner product in Eq.\ref{innerprod}, we can define notions of orthogonality and projection. 
In the case of the two layer network ($L=2$) with weights $\bm{U}$ and $\bm{W}$, the inner product for two vectors $\bm{e} = (\bm{E_U},\bm{E_W})$ and $\bm{f} = (\bm{F_U},\bm{F_W})$ becomes:
 \begin{align} \label{innerprod2}
    \bm{e}^T \bm{f} = \tr(\bm{E_W}^T \bm{F_W}) + \tr(\bm{E_U}^T \bm{F_U})
 \end{align}
The normal space can be defined as the space orthogonal to the tangent space. This is the same as the eigenspace of the Hessian corresponding to positive eigenvalues and was described in Section \ref{HeigenspaceM} (similarly, vectors in the tangent space corresponds to Hessian eigenvectors with zero eigenvalue, i.e. $\bm{H}\bm{t} = 0$). Finally, the inner product of an arbitrary vector $\bm{g}$ with vector $\bm{t}$ in the tangent space can be found from:
\begin{align} \label{tangpr1}
    \bm{t}^T\bm{g} = \trace (\bm{T_W}^T\bm{G_W}) + \trace (\bm{T_U}^T\bm{G_U}) = \trace (\bm{T_W}(\bm{G_W}^T+\bm{G_U}))
\end{align}
where $\bm{T_W}$ takes an appropriate form from Lemma $\ref{lemmatangent}$. A natural basis for the tangent space can be constructed from matrices $\bm{\Omega}_{rs}:=(\bm{v}_r\bm{v}_s^T - \bm{v}_s\bm{v}_r^T)/\sqrt{2\omega_r\omega_s(\omega_r+\omega_s)}$, for $r,s \in [1,n]$ and $r>s$, as well as matrices $\bm{K}_{uv}$ with components $[\bm{K}_{uv}]_{ij} = \delta_{ij}^{uv}$ ($u,i \in [1,n]$, $v,j \in [1,p-n]$). Recall that $\omega_i = 1-\gamma/s_i$ and $(\bm{v}_i,s_i)$ are eigenvector/eigenvalue pairs of the input covariance, and the normalization constants ensure $\bm{t}^T\bm{t}=1$. In Section \ref{prooftheoremphi}, we use these results to project the gradient onto the tangent space.

\section{Calculation of the Fluctuations}
In this section, we present additional results and derivations for the fluctuations in the normal space. 
\subsection{Derivation of the SDE for normal deviations}
\label{derivationSDE}
Here, we will derive the stochastic differential equation (SDE) for $\bm{\theta}_N$ in Eq.\ref{SDEthetaN}.
Without loss of generality, at a given time, we take $\tilde{\bm{\theta}}$ to be the closest point to $\bm{\theta}$, and also put the origin on it.
Projecting both sides of the SGD equation (Eq.\ref{sgd}) to the normal space of the manifold, and by defining $\bm{g}_N(\bm{x};\bm{\theta}_N) = \Pi_N(\bm{g}(\bm{x};\bm{\theta}_N))$ to be the normal projection of the gradient, we have: 
\begin{align}
\Delta \bm{\theta}_N &= -\eta \bm{g}_N
=  -\eta \langle \bm{g}\rangle_x  - \eta\bm{e}  \nonumber
\end{align}
where $ \bm{e} = \bm{g}_N - \langle \bm{g}\rangle_x$ is a random variable with zero mean, and $\langle \bm{g}_N\rangle_x = \langle \bm{g}\rangle_x$ is the average gradient is in the normal direction. The above is the update in one training step which we take to be equivalent to the continuous time-difference $\Delta t_{step} \equiv \eta$. Over large timescales ($\Delta t = n\eta$, with $n \gg 1$), the update equation becomes:
\begin{align}
\Delta\bm{\theta}_N &= -n\eta \langle \bm{g} \rangle_x - \bm{e'} \quad \text{(n steps)}
\end{align}
Here $\bm{e'} = \eta\sum_{i=1}^{n}\bm{e}_i$ is a random variable with $\langle \bm{e'} \rangle = 0$ and $\text{var}(\bm{e'}) = n\eta^2\text{cov}(\bm{e}) = \Delta t \,\eta \,\text{cov}(\bm{e})$. As a result, we can approximate it as $\bm{e'} \sim \sqrt{\eta\,\text{cov}(\bm{e})} \bm{B}_{\Delta t}$ where $\bm{B}_t$ is a multi-dimensional Brownian motion (Wiener process) with properties such as independent Gaussian increments and a variance proportional to the time difference, i.e. $\bm{B}_{t+u}-\bm{B}_{t} \sim \mathcal{N}(\bm{0},u\bm{I}_{2np})$. In the limit of $\Delta t \rightarrow 0$, we have the following continuous stochastic differential equation (SDE):
\begin{align}
d\bm{\theta}_N &= -\langle \bm{g} \rangle_x dt + \sqrt{\eta}\bm{C} d\bm{B}_t \quad \text{(continuous)}
\end{align}
where $\bm{C} \equiv \sqrt{\text{cov}(\bm{e})}$.
Since $\bm{e}$ in general is a function of $\bm{\theta}_N$, we can write $\bm{C}(\bm{\theta}_N) = \bm{C}(0) + \mathcal{O}(\theta_N)$. 
Near the manifold we also have $\langle \bm{g} \rangle_x = \bm{H}\bm{\theta}_N + \mathcal{O}(\theta_N^2)$. As a first approximation, we replace $\bm{C}(\bm{\theta}_N)$ by $\bm{C}(0)$, which means the driver of the stochasticity in the equation becomes the gradient on the manifold, i.e. $\tilde{\bm{g}}(\bm{x}): = \bm{g}_N(\bm{x};\bm{0})$.
Note that as the next section shows, under the stationary solution $\theta_N \sim \sqrt{\eta}$, which justifies the approximation in the second term of the equation for small $\eta$ (see also \citet{mandt2017stochastic}).

\subsection{Solution of the SDE for normal deviations}
\label{SDEsolution}
The solution to the SDE in Eq.\ref{SDEthetaN} (aka Ornstein Uhlenbeck process) has been discussed previously \cite{gardiner1985handbook}. For completeness, we provide it here. 
The process was defined by:
\begin{align}
d\bm{\theta}_N(t) = -\bm{H}\bm{\theta}_N(t) dt + \sqrt{\eta}\bm{C} d\bm{B}_t
\label{OU2}
\end{align}
We first perform the following change of variable:
\begin{align} \label{eqchvar}
\bm{\theta}'_N(t) := e^{\bm{H}t}\bm{\theta}_N(t) \, \rightarrow \, \bm{\theta}_N(t) = e^{-\bm{H}t}\bm{\theta}'_N(t).
\end{align}
By taking the differential of the above and replacement from Eq.\ref{OU2} we get:
\begin{align}
d\bm{\theta}'_N(t) &= \bm{H}e^{\bm{H}t}\bm{\theta}_N(t) dt + e^{\bm{H}t}(-\bm{H}\bm{\theta}_N(t) dt + \sqrt{\eta}\bm{C} d\bm{B}_t) = \sqrt{\eta} e^{\bm{H}t}\bm{C} d\bm{B}_t
\end{align}
After integration (Ito integral) and substitution from \ref{eqchvar}, we have:
\begin{align}
\bm{\theta}'_N(t) = \bm{\theta}'_N(0) + \sqrt{\eta}\int_0^t e^{\bm{H}u}\bm{C} d\bm{B}_u \rightarrow \bm{\theta}_N(t) = e^{-\bm{H}t}\bm{\theta}_N(0) +\sqrt{\eta} \int_0^t e^{-\bm{H}(t-u)}\bm{C} d\bm{B}_u 
\end{align}
For a given initial condition, the mean is:
\begin{align}
\langle \bm{\theta}_N(t) \rangle = e^{-\bm{H}t}\bm{\theta}_N(0)
\end{align}
which shows the stationary mean ($t \rightarrow \infty$) is zero. Correspondingly, the correlation function is:
\begin{align}
\text{corr}(\bm{\theta}_N(t),\bm{\theta}_N(s)) &= \langle[\bm{\theta}_N(t)-\langle\bm{\theta}_N(t)\rangle][\bm{\theta}_N(s)-\langle\bm{\theta}_N(s)\rangle]^T\rangle \ \\ &= \eta\langle  \int_0^t e^{-\bm{H}(t-u)}\bm{C} d\bm{B}_u  [\int_0^s e^{-\bm{H}(s-v)}\bm{C} d\bm{B}_v]^T \rangle \nonumber \\
&= \eta \int_0^{\min(t,s)} e^{-\bm{H}(t-u)} \bm{C}\bm{C}^T e^{-\bm{H}^T(s-u)} du \nonumber
\end{align}
where in the last line we used a property of Ito integral. Now recall the change of variable from the main text $\bm{\rho} = \bm{N}^T\bm{\theta}_N$ (Eq.\ref{rhodef}), which means $\text{corr}(\bm{\rho}(t),\bm{\rho}(s)) = \bm{N}^T\text{corr}(\bm{\theta}_N(t),\bm{\theta}_N(s))\bm{N}$. Additionally, we have the decomposition $\bm{H} = \bm{N}\bm{\lambda}\bm{N}^T$, where $\bm{\lambda} = \text{diag}(\{\lambda_k\})$. Replacing these and $\bm{C}\bm{C}^T = \langle \tilde{\bm{g}}\tilde{\bm{g}}^T \rangle_x$ in the above, the integral can be calculated per component. Further, taking $t,s \to \infty$ yields the correlation function for the stationary solution:
\begin{align}
\langle \rho_k(t)\rho_l(s) \rangle = \frac{\eta\langle {\bm{n}}_k^T\tilde{\bm{g}}\,  {\bm{n}}_l^T\tilde{\bm{g}} \rangle_x}{\lambda_k+\lambda_l}e^{-\lambda_k(t-s)}, \, (t\geqslant s).
\end{align}
The covariance equation in the main text (Eq.\ref{rhocovariance}) results from the above under $t=s$.  

\subsection{Projection of the gradient onto the normal space (Proof of Proposition \ref{propH})} \label{proofpropH} 
%\begin{proof}
In Section \ref{HeigenspaceM} we found the normal space which consisted of two sets of solutions (Eqns. \ref{Hsol1} and \ref{Hsol2}). Here we show that $\tilde{\bm{g}}(\bm{x})$ has no projection the first set. Recall that for the first solution we had: $\bm{n} = (-\bm{N_W}^T,\bm{N_W})$ where $\bm{N_W}\Wt^T = \Wt\bm{N_W}^T$. Using $\tilde{\bm{g}}(\bm{x}) = (\Wt^T\bm{Z_x},\bm{Z_x}\Wt)$ 
from Eq.\ref{gM}, and the definition of the inner product (Eq.\ref{innerprod2}), we have:
\begin{align}
\text{Set 1:}\quad \bm{n}^T \tilde{\bm{g}}(\bm{x}) &= \tr(\bm{N_W}^T\bm{Z_x}\Wt) - \tr(\bm{N_W}\Wt^T\bm{Z_x}) = \tr(\bm{N_W}^T\bm{Z_x}\Wt) - \tr(\Wt\bm{N_W}^T\bm{Z_x}) = 0
\nonumber
\end{align}
where in the last equality we used the permutation property of the trace. This leaves us with the second set of solutions which consisted of $\bm{n}_{ij} = (\Wt^T\bm{Z}_{ij},\bm{Z}_{ij}\Wt)$ for $i,j \in [1,n]$ (Eq.\ref{Hsol2i}). The projection on these eigenvectors can be found similarly using the definition of inner product. By defining the coefficients
\begin{align}
S^{ij}_{rs} := \frac{\kappa_{ij}}{s_r} + \frac{1}{s_s}\,(\text{noting} \, \kappa_{ii} \equiv 0), \quad \text{and} \quad \omega_i := 1 - \frac{\gamma}{s_i},
\label{Sij}
\end{align}
(see Section \ref{HeigenspaceM} for definitions of $\kappa_{ij}$ and $C_{ij}$), we will have:
\begin{align} \label{gMproj}
  \text{Set 2:}\quad  \bm{n}_{ii}^T \tilde{\bm{g}} = \gamma\sqrt{2\omega_i}(1 - \frac{x_i^2}{s_i}), \quad
    \bm{n}_{ij}^T \tilde{\bm{g}}_{(i\neq j)} = -\gamma (\omega_{i}+\omega_{j})C_{ij}S^{ij}_{ij}x_ix_j,
\end{align}
where $x_i := \bm{v}_i^T\bm{x}$.
Hence the projections on set 2 are in general non-zero. This completes the proof.
%\end{proof}

\subsection{Components of the fluctuation matrix} 
\label{calcCompFluc} 
By replacing the projections from Eq.\ref{gMproj} in Eq.\ref{rhocovariance}, the components of the covariance matrix can be calculated as:
\begin{align} \label{rhoijpq}
    &\langle \rho_{ii}\rho_{pp}\rangle
    = \frac{\eta\gamma^2\sqrt{\omega_i\omega_p}}{s_i+s_p-2\gamma}(\frac{\langle x_i^2x_p^2\rangle_x}{s_is_p} - 1), \qquad \langle \rho_{ii}\rho_{pq}\rangle_{(p\neq q)}
    =\langle \rho_{ij}\rho_{pq}\rangle_{(i\neq j)} =0\\
    &\langle \rho_{ij}\rho_{pq}\rangle_{(i\neq j,p\neq q)}
     = \frac{\eta\gamma^2 (\omega_{i}+\omega_{j})(\omega_{p}+\omega_{q})}{\lambda_{ij} + \lambda_{pq}}C_{ij}C_{pq}S^{ij}_{ij}S^{pq}_{pq}\langle x_ix_jx_px_q \rangle_x \nonumber
\end{align}
where $i,j,p,q \in [1,n]$.
\subsection{Fluctuation of the representation norm} \label{flucrepnorm}
Consider $\bm{h}_s = \Ut\bm{v}_s$ to be the representation of stimulus $\bm{v}_s$. According to Proposition \ref{propH}, displacement from the manifold along vector $\bm{\rho} = \sum \rho_k \bm{n}_k$ in the normal space changes the first layer weight by $\Delta\bm{U} = \Wt^T\sum_{i,j} \rho_{ij}C_{ij}(\kappa_{ij}\bm{v}_i\bm{v}_j^T + \bm{v}_j\bm{v}_i^T)$. This correspondingly changes the representation by $\Delta \bm{h}_s = \Delta\bm{U}\bm{v}_s$.
The change in the representation norm can therefore be calculated as:
\begin{align}
\hat{\bm{h}}_s^T\Delta \bm{h}_s 
= \sqrt{\omega_s}\bm{v}_s^T \sum_{i,j=1}^{n} \rho_{ij}C_{ij}(\kappa_{ij}\bm{v}_i\bm{v}_j^T + \bm{v}_j\bm{v}_i^T)\bm{v}_s = \frac{\rho_{ss}}{\sqrt{2}}
\end{align}
where we used $\Wt\Wt^T\bm{v}_s = \omega_s\bm{v}_s$ and $C_{ss} = 1/\sqrt{2\omega_s}$. The variance in the norm becomes:
\begin{align}
\text{var}(|\bm{h}_s|) = \frac{\langle \rho_{ss}^2 \rangle}{2}
\label{varh} 
\end{align}
which is the first part of Eq.\ref{hvar} in the main paper. The right hand side of Eq.\ref{hvar} results simply by taking $i=p=s$ in the first equation of Eq.\ref{rhoijpq} to calculate $\langle \rho_{ss}^2 \rangle$.

\newpage
\section{Tangential Projection of the Gradient (Proof of Theorem \ref{theoremphi})} \label{prooftheoremphi}
%\begin{proof}
We first approximate the gradient near the manifold to the first order of $\bm{\rho}$. We assume we are at point $\tilde{\bm{\theta}} + \rho_k\bm{n}_k$, where $\tilde{\bm{\theta}}$ lies on the manifold, and $\bm{n}_k = (\Wt^T\bm{Z}_k,\bm{Z}_k\Wt)$ is a Hessian eigenvector (Eq.\ref{nk}). The gradient at this point can be calculated by replacing $\bm{W} = (\bm{I} + \rho_k \bm{Z}_k)\Wt$ and $\bm{U} = \Wt^T(\bm{I} + \rho_k\bm{Z}_k)$ in Eq.\ref{eqnGradient}. For $\rho_k \ll 1$, we will ignore second and higher order terms, to have:
\begin{align} 
&\hspace{-2em}\bm{g}(\bm{x};\tilde{\bm{\theta}} + \rho_k\bm{n}_k) = (\bm{G_U},\bm{G_W}); \label{gradapprox}\\
    \bm{G_W} &= (\bm{W}\bm{U}-\bm{I}) \bm{x}\bm{x}^T\bm{U}^T +\gamma\bm{W} \nonumber\\ &= [(\bm{I}+\rho_k\bm{Z}_k)\Wt\Wt^T(\bm{I}+\rho_k\bm{Z}_k)-\bm{I}] \bm{x}\bm{x}^T(\bm{I}+\rho_k\bm{Z}_k^T)\Wt +\gamma(\bm{I}+\rho_k\bm{Z}_k)\Wt + \mathcal{O}(\rho_k^2)  \nonumber \\
    &= \gamma(\bm{I} -\Sxi\xxt)\Wt + \rho_k[(\bm{Z}_k\Wt\Wt^T+ \Wt\Wt^T\bm{Z}_k)\bm{x}\bm{x}^T -\gamma\Sxi\bm{x}\bm{x}^T\bm{Z}_k^T + \gamma\bm{Z}_k]\Wt + \mathcal{O}(\rho_k^2) \nonumber\\
    \bm{G_U} &=  \bm{W}^T(\bm{W}\bm{U}-\bm{I})\bm{x}\bm{x}^T +\gamma\bm{U} \nonumber\\
    &= \Wt^T(\bm{I}+\rho_k\bm{Z}_k^T)[(\bm{I}+\rho_k\bm{Z}_k)\Wt\Wt^T(\bm{I}+\rho_k\bm{Z}_k)-\bm{I}] \bm{x}\bm{x}^T +\gamma\Wt^T(\bm{I}+\rho_k\bm{Z}_k) + \mathcal{O}(\rho_k^2) \nonumber \\
    &=\gamma\Wt^T(\bm{I} - \Sxi\bm{x}\bm{x}^T) + \rho_k\Wt^T[(\bm{Z}_k\Wt\Wt^T + \Wt\Wt^T\bm{Z}_k)\bm{x}\bm{x}^T-\gamma\bm{Z}_k^T\Sxi\bm{x}\bm{x}^T + \gamma\bm{Z}_k] + \mathcal{O}(\rho_k^2) \nonumber
\end{align}
As was shown Section \ref{secTangent}, the dependency of the projection on the gradient is via $\bm{G_W}^T + \bm{G_U}$, which can be calculated as:
\begin{align}
 \bm{G_W}^T + \bm{G_U} &= -\gamma\rho_{ij}\Wt^T(\bm{Z}_{ij}\xxt\Sxi + \bm{Z}_{ij}^T\Sxi\xxt) + 2\Wt^T\text{Sym}(\bm{A})+ \mathcal{O}(\rho_{ij}^2) \nonumber \\ &=  -\gamma\rho_{ij} \Wt^T C_{ij}\sum_{q \in [1,n]} (S^{ij}_{qj}x_{j}x_q \bm{v}_{i}\bm{v}_q^T + S^{ij}_{iq}x_{i}x_q \bm{v}_{j}\bm{v}_q^T) + 2\Wt^T\text{Sym}(\bm{A}) + \mathcal{O}(\rho_{ij}^2).
 \label{Gz}
\end{align}
In the above, $\bm{A} = \gamma(\bm{I}_n-\bm{x}\bm{x}^T\Sxi) + \rho_{ij}[(\bm{Z}_{ij}\Wt\Wt^T + \Wt\Wt^T\bm{Z}_{ij})\bm{x}\bm{x}^T + \gamma\bm{Z}_{ij}]$, and $S^{ij}_{rs}$ are coefficients defined in Eq.\ref{Sij}. In the second line, we replaced $\bm{Z}_{k} \equiv \bm{Z}_{ij} = C_{ij}(\kappa_{ij}\bm{v}_{i} \bm{v}_{j}^T + \bm{v}_{j}\bm{v}_{i}^T)$ (Eq.\ref{nk}), $\bm{x} = \sum_{q \in [1,n]} x_q \bm{v}_q$, and used $\Sxi\bm{v}_q = s_q^{-1}\bm{v}_q$ and $\bm{v}_p^T\bm{v}_q = \delta_p^q$.

The inner product of the gradient on a vector in the tangent space becomes $\bm{t}^T\bm{g} = \trace (\bm{T_W}(\bm{G_W}^T+\bm{G_U}))$ (Eq.\ref{tangpr1}), where from Lemma \ref{lemmatangent} we have $\bm{T_W} = \Wt\Wt^T\bm{\Omega}\Wt + \bm{K}\Wt_{\bot}$. Replacing the $\bm{G_W}^T + \bm{G_U}$ term from Eq.\ref{Gz}, we first see that the gradient has no projection on the second term of the tangent vector containing $\bm{K}$, as we have $\Wt_{\bot}\Wt^T = 0$. Further, the contribution of the term $2\Wt\text{Sym}(\bm{A})$ in Eq.\ref{Gz} also disappears as the trace of the product of symmetric and skew-symmetric matrices is zero. This leaves us with calculating the projection on tangent vectors corresponding to $\bm{\Omega}$. As described in Section \ref{secTangent}, the associated basis vectors are $\bm{t}_{rs} = (\bm{T_W}^T,\bm{T_W})$, where $\bm{T_W}=\Wt\Wt^T\bm{\Omega}_{rs}\Wt$, and $\bm{\Omega}_{rs}:=(\bm{v}_r\bm{v}_s^T - \bm{v}_s\bm{v}_r^T)/\sqrt{2\omega_r\omega_s(\omega_r+\omega_s)}$. The projection on $\bm{t}_{rs}$ becomes:
\begin{align}
\bm{t}_{rs}^T\bm{g}(\bm{x};\tilde{\bm{\theta}}\!+\! \rho_k\bm{n}_k) &= \trace (\Wt\Wt^T \bm{\Omega}_{rs}\Wt(\bm{G_W}^T + \bm{G_U}))  \nonumber
\\
&= \frac{\gamma \rho_{ij}C_{ij}\sqrt{\omega_s\omega_r}}{\sqrt{2(\omega_s+\omega_r)}}[x_j(S^{ij}_{sj}x_s \delta_i^r - S^{ij}_{rj}x_r \delta_i^s) + x_i(S^{ij}_{is}x_s \delta_j^r - S^{ij}_{ir}x_r \delta_j^s)] + \mathcal{O}(\rho_{ij}^2) \nonumber
\\ &= \rho_{ij}\sqrt{2(\omega_s + \omega_r)}\G_{ij}^{s,r}(\bm{x}) + \mathcal{O}(\rho_{ij}^2)\label{gtrs}
\end{align}
where in the last line we defined the rank-3 tensor $\G$ with components that are:
\begin{align}
\G_{ij}^{s,r}(\bm{x}) = -\G_{ij}^{r,s} (\bm{x}) =\frac{\gamma C_{ij}\sqrt{\omega_r\omega_s}}{2(\omega_r + \omega_s)} [x_j(S^{ij}_{sj}x_s \delta_i^r - S^{ij}_{rj}x_r \delta_i^s) + x_i(S^{ij}_{is}x_s \delta_j^r - S^{ij}_{ir}x_r \delta_j^s)].
\end{align}
(note that the lower index [$k$] is a composite index and we use it interchangeably with $ij$, i.e. $\G_{k}^{s,r} \equiv \G_{ij}^{s,r}$). The projection of the gradient on the tangent space becomes $\bm{g}_T = \sum_{r>s} (\bm{t}_{rs}^T\bm{g})\,\bm{t}_{rs} = (\bm{G}_{\bm{U}_T},\bm{G}_{\bm{U}_T}^T)$, where:
\begin{align} \label{Gutk}
\bm{G}_{\bm{U}_T}(\bm{x};\tilde{\bm{\theta}}\! +\!\rho_{ij}\bm{n}_{ij})  &= \rho_{ij}\Wt^T \sum_{r>s}\G_{ij}^{s,r} (\sqrt{\frac{\omega_r}{\omega_s}}\bm{v}_s\bm{v}_r^T-\sqrt{\frac{\omega_s}{\omega_r}}\bm{v}_r\bm{v}_s^T) + \mathcal{O}(\rho_{ij}^2)
\\ &= \Wt^T (\Wt\Wt^T)^{-\frac{1}{2}} (\rho_{ij}\G_{ij}^{:,:})(\Wt\Wt^T)^{\frac{1}{2}} + \mathcal{O}(\rho_{ij}^2)\nonumber
\end{align}
(In the second line above we used $\G_{ij}^{:,:} = \sum_{s,r} \G_{ij}^{r,s}\bm{v}_s\bm{v}_r^T$).
Due to linearity, the tangential projection at point $\bm{N}\!\bm{\rho} = \sum_k \rho_k \bm{n}_k$ can be found by summing over $k$, leading to Eq.\ref{gT} in the main text:
\begin{align} \label{gradUt}
\bm{G}_{\bm{U}_T}(\bm{x};\tilde{\bm{\theta}}\! +\!\bm{N}\bm{\rho}) &= \sum_{k} \bm{G}_{\bm{U}_T}(\bm{x};\tilde{\bm{\theta}} \!+\!\bm{\rho}_k\bm{n}_k) = \Wt^T (\Wt\Wt^T)^{-\frac{1}{2}} (\sum_k{\rho_k}\bm{\mathcal{G}}_{k}^{:,:})(\Wt\Wt^T)^{\frac{1}{2}} + \mathcal{O}(|\bm{\rho}|^2)
\end{align}

Next, we will find the effect of the tangential gradient on the representations. For a given stimulus $\bm{s}$, the change in its representation due to the tangential gradient is $\Delta \bm{h} = -\eta\bm{G}_{{\bm{U}}_T} \bm{s}$. Replacing $\bm{G}_{{\bm{U}}_T}$ from Eq.\ref{gradUt}, and using $\bm{h} = \bm{U}\bm{s} \approx \Wt^T\bm{s}$ (for small deviations from the manifold), we have:
\begin{align}
\bm{h}^T\Delta\bm{h} = -\eta\bm{s}^T (\Wt\Wt^T)^{\frac{1}{2}} (\sum_k{\rho_k}\bm{\mathcal{G}}_{k}^{:,:})(\Wt\Wt^T)^{\frac{1}{2}} \bm{s} = 0
\end{align}
The second equality results from the fact that the quadratic form of a skew-symmetric matrix ($\G_k^{:,:}$) is zero. This shows that, as expected, the tangential projection of the gradient is equivalent to a small rigid-body rotation of the representations around the origin. Further, $\Wt^T$ on the left hand side of Eq.\ref{gradUt} makes $\Delta \bm{h}$ to be a linear combination of the columns of $\Wt^T$, keeping the representations in the column-space of $\Wt^T$ (note also that the column-space of $\Wt^T$ stays fixed over a tangential update, which is again due to the $\Wt^T$ term on the left hand side of Eq.\ref{gradUt}). To characterize this rotation, we measure its effect on the $n$ representations vectors. Specifically, we define $\Delta {\varphi}_{sr}$ to be the the pairwise angular displacement of representation $\bm{h}_s \approx \Wt^T\bm{v}_s$ toward $\bm{h}_r \approx \Wt^T\bm{v}_r$:
\begin{align}
    \Delta {\varphi}_{sr} := \frac{\tilde{\bm{h}}_r^T\Delta\bm{h}_s}{|\tilde{\bm{h}}_r||\tilde{\bm{h}}_s|} 
    = \eta \sum_k{\rho_k}\G_{k}^{s,r} + \mathcal{O}(|\bm{\rho}|^2),
    \label{phisr}
\end{align}
where we replaced $\Delta\bm{h}_s$ from the above and used $|\bm{h}_s| = \sqrt{\omega_s}$, and $\G_{k}^{r,s} = \bm{v}_r^T\G_{k}^{:,:}\bm{v}_s$. This is Eq.\ref{dphi} in the main text.
%\end{proof} 

\section{Derivation of the Diffusion Coefficients for the Case with a Frequent Stimulus} \label{derfreqstimulus}
Here we derive approximate analytical solutions for the diffusion coefficients for the case with a frequent stimulus.
\subsection{Approximation of the diffusion summation} \label{approxsum}
Since in the $D_{sr}$ summation in Eq.\ref{Dsr}, $k \equiv (ij)$ and $l \equiv (pq)$ are composite indices with $i,j,p,q \in [1,n]$, the summation is in general performed over $n^4$ terms (first line of Eq.\ref{Dsrapprox} below). However, as we'll see next, most of the terms are zero and under different limits, we can further constrain it to a subset of indices. 

\underline{$n \gg 1$}: 
If we consider the equation for $\G_{ij}^{r,s}$ for a given $r$ and $s$ (Eqn, \ref{G}), we see that it is non-zero only when at least either $i$ or $j$ are equal to $r$ or $s$. For a fixed $r$ and $s$, looping over $i,j \in [1,n]$ leads to a maximum of two non-zero components over the $\G_{ii}^{r,s}$ terms, and up to $\approx 4n$ non-zero components over the $\G_{ij(i \neq j)}^{r,s}$ terms. Hence, for large $n$,
the terms corresponding to the first set can be ignored. This leaves us with a summation over indices $(ij,pq)$ for $i\neq j$ and $p \neq q$, which consists of up to $\approx 16n^2$ non-zero terms.
Additionally, both $\Gbar_{ij,pq}^{r,s}$ and $\langle \rho_{ij}\rho_{pq}\rangle$ are proportional to terms like $\langle x_{r/s}x_{i/j}x_{r/s}x_{p/q}\rangle_x$, which are zero if an odd number of indices are equal. This leaves us with a summation over indices of $(ij,ij)$ and $(ij,ji)$, where  $(i,j) \in \mathbb{C}_{[ij]} \cup \mathbb{C}_{[ji]}$, for set $\mathbb{C}_{[ij]} := \{(i,j)| i \in \{r,s\}, j\in [1,n] \backslash\{r,s\} \}$ (this is shown in the second line of Eq.\ref{Dsrapprox} below). Since $\mathbb{C}_{[ij]}$ contains $\approx 2n$ terms, the whole summation will have up to $\approx 8n$ terms.

\underline{$\alpha \ll 1$}:
For small $\alpha$, we can further limit the summation to indices $(ij,ij)$ for $i > j$. This is because the first order correction to the summands in Eq.\ref{Dsr} is: $\approx \langle \rho_k\rho_l\rangle|_{\alpha=0} \Gbar^{r,s}_{k,l}|_{\mathcal{O}(\alpha)} + \langle\rho_k\rho_l\rangle|_{\mathcal{O}(\alpha)} \Gbar^{r,s}_{k,l}|_{\alpha=0}$, and so we can limit the sum to indices for which $\langle\rho_k\rho_l\rangle|_{\alpha=0}$ or $\Gbar^{r,s}_{k,l}|_{\alpha=0}$ are non-zero. This corresponds to the previous Gaussian stimuli case and the terms were calculated in Section \ref{whitenoise} of the main text.
To summarize:
\begin{align} \label{Dsrapprox}
    {D}_{sr} &=(\eta^2/2) \sum_{i,j,p,q=1}^{n} \langle \rho_{ij} \rho_{pq} \rangle \,\Gbar_{ij,pq}^{r,s} 
    \nonumber \\& \approx (\eta^2/2)\sum_{(i,j) \in \mathbb{C}_{[ij]} \cup \mathbb{C}_{[ji]}} (\langle \rho_{ij}^2 \rangle \,\Gbar_{ij,ij}^{r,s} + \langle \rho_{ij} \rho_{ji} \rangle \,\Gbar_{ij,ji}^{r,s} ) \qquad  n\gg 1 \nonumber \\ 
    &\approx (\eta^2/2) \sum_{(i,j) \in (\mathbb{C}_{[ij]} \cup \mathbb{C}_{[ji]}) \wedge\, (i>j)} \langle \rho_{ij}^2 \rangle \,\Gbar_{ij,ij}^{r,s} \qquad\qquad \alpha \ll 1, n\gg 1
\end{align}
We see that for a given $s$ and $r$, we will have up to $\approx 2n$ terms in the summation under the $n\gg 1$ and $\alpha \ll 1$ limit.

To fully characterize the diffusion in the representation space, we only need to find two coefficients $D_{ab}$ ($=D_{ba}$) and $D_{bc}$ ($= D_{cb}$), where $a=1$, $b \neq c \in [2,n]$. Without loss of generality, we take $b=2$ and $c=3$. We also take $d=4$ to be another index within the background subspace. In the next two sections, we calculate these coefficients under two limits.

\subsection{Calculation for large $n$ and small $\alpha$} \label{Gcalcfreqstim}
For small $\alpha$, $D_{ba}$ and $D_{bc}$ can be found by replacement in the last line of Eq.\ref{Dsrapprox}:
\begin{align*}
&D_{ba} \approx \frac{n\eta^2}{2}(\langle\rho_{db}^2\rangle\Gbar_{db,db}^{a,b} + \langle \rho_{da}^2\rangle \Gbar_{da,da}^{a,b}) \qquad  (\alpha \ll 1, n \gg 1)
\nonumber \\
&D_{bc} \approx \frac{n\eta^2}{2}(\langle \rho_{db}^2 \rangle \Gbar_{db,db}^{c,b} + \langle \rho_{dc}^2 \rangle \Gbar_{dc,dc}^{c,b}) \qquad  (\alpha \ll 1, n \gg 1)
\end{align*}
The relevant components of tensor $\G$ can be calculated from Eq.\ref{G} for the case of frequent stimulus. For $\gamma \ll 1$ we have $\omega_i \approx 1$, $\kappa_{ij(i>j>1)}=1$, $C_{ij (i\neq j >1)} \approx 1/2$, $\langle x_i^2x_j^2\rangle_{x (i \neq j)} = 1-\alpha$, and for small $\alpha$: $\kappa_{i1(i>1)} \approx 1-\alpha/\gamma$. After replacement and keeping the terms to the first order of $\alpha$, the components of $\G$ and $\Gbar$ become:
\renewcommand{\arraystretch}{2}
\begin{center}
\[
\begin{array}{ll}
 \G^{a,b}_{db} = \frac{\gamma}{8}\frac{x_ax_d}{s_d}(1 + \frac{s_d}{s_a}) & \Gbar_{db,db}^{a,b} =  \frac{\gamma^2}{64}\frac{\langle x_a^2x_d^2\rangle_x }{s_d^2}(1 + \frac{s_d}{s_a})^2 = \frac{\gamma^2}{16} \nonumber \\ 
 \G^{a,b}_{da} = \frac{-\gamma}{4}\frac{x_bx_d}{s_b}\frac{1 + \kappa_{da}}{\sqrt{2(1 + \kappa_{da}^2})} \approx  \frac{-\gamma}{4}\frac{x_bx_d}{s_d} & \Gbar^{a,b}_{da,da} = \frac{\gamma^2}{16}\frac{\langle x_b^2x_d^2\rangle_x}{s_d^2} = \frac{\gamma^2}{16}(1+\alpha)  \nonumber \\ 
 \G^{c,b}_{db} = \frac{\gamma}{4} \frac{x_cx_d}{s_d}, & \Gbar^{c,b}_{db,db} = \frac{\gamma^2}{16} \frac{\langle x_c^2x_d^2 \rangle_x}{s_d^2} = \frac{\gamma^2}{16}(1+\alpha)\nonumber\\
 \G^{c,b}_{dc}  = \frac{-\gamma }{4} \frac{x_bx_d}{s_d}, &  \Gbar^{c,b}_{dc,dc} =  \frac{\gamma^2 }{16} \frac{\langle x_b^2x_d^2 \rangle_x}{s_d^2} = \frac{\gamma^2}{16}(1+\alpha)
 \nonumber
\end{array}
\]
\end{center}
Replacement of the above and the fluctuations from Eq.\ref{rhoijpq} leads to Eq.\ref{Dsrsmallalpha} of the main paper.

\subsection{Calculation for large $n$ and $\alpha \gg \gamma$} \label{CalcfreqstimLargeN}
Here, we derive closed form solutions under the assumptions of $\gamma \ll 1$ and $s_a,s_b,s_a-s_b \gg \gamma$ (equivalently $\alpha \gg \gamma$). Using the second line of Eq.\ref{Dsrapprox}, the pairwise diffusion coefficients become:
\begin{align}
D_{ba} &\approx \frac{n\eta^2} {2}(\langle\rho_{db}^2\rangle\Gbar_{db,db}^{a,b} + \langle \rho_{da}^2\rangle \Gbar_{da,da}^{a,b} + \langle \rho_{ad}^2\rangle \Gbar_{ad,ad}^{a,b} + 2\langle \rho_{ad}\rho_{da}\rangle \Gbar_{ad,da}^{a,b}) \quad && (n \gg 1)
\nonumber\\
& \approx \frac{\eta^3\gamma^4n}{2} [\frac{\langle x_a^2x_d^2\rangle\langle x_b^2x_d^2\rangle(\frac{1}{s_a} + \frac{1}{s_b})^2}{64s_b^3} + \frac{\langle x_a^2x_d^2\rangle\langle x_b^2x_d^2\rangle}{64s_a^2s_b^3} +  \frac{\langle x_a^2x_d^2\rangle\langle x_b^2x_d^2\rangle}{64s_as_b^4} +  \frac{\langle x_a^2x_d^2\rangle\langle x_b^2x_d^2\rangle}{16(s_a+s_b)s_as_b^3} ] &&(\alpha \gg \gamma) \nonumber \\
&= \frac{\eta^3\gamma^4n}{128} \frac{\langle x_a^2x_d^2\rangle\langle x_b^2x_d^2\rangle} {s_b^5} [1 + \frac{3s_b}{s_a} + \frac{2s_b^2}{s_a^2}  + \frac{4s_b^2}{s_a(s_a+s_b)}] \nonumber \\
&= \frac{\eta^3\gamma^4n}{16(1-\alpha)^3}\frac{1}{16}[1 + 3(1-\alpha) + 2(1-\alpha)^2  + \frac{4(1-\alpha)^2}{2-\alpha}] = \frac{\eta^3\gamma^4n}{16(1-\alpha)^3} [\frac{16-28\alpha+15\alpha^2-2\alpha^3}{8(2-\alpha)}] \nonumber
\\
\nonumber\\
D_{bc} &\approx \frac{n\eta^2}{2}(\langle \rho_{db}^2 \rangle \Gbar_{db,db}^{c,b} + \langle \rho_{dc}^2 \rangle \Gbar_{dc,dc}^{c,b}) \qquad\qquad && (n \gg 1) \nonumber \\
&= \frac{\eta^3\gamma^4n}{16} \frac{\langle x_b^2x_d^2 \rangle\langle x_c^2x_d^2 \rangle}{s_b^5(1-\frac{\gamma}{s_b})} \nonumber \\ &= \frac{\eta^3\gamma^4n}{16(1-\alpha)^3(1-\frac{\gamma}{s_b})}
\overset{\gamma \ll 1}{\approx} \frac{\eta^3\gamma^4n}{16(1-\alpha)^3}
\end{align} 
In the above, we replaced the appropriate quantities of $\langle \rho_{ij}\rho_{pq}\rangle$ and $\Gbar_{ij,pq}^{r,s}$, calculated from Eq.\ref{rhoijpq} and Eq.\ref{G} respectively (see below). In deriving these terms, we used $\omega_i \approx 1$, $\langle x_i^2x_j^2\rangle_{x (i \neq j)} = 1-\alpha$, $\lambda_{db} = 2(s_b - \gamma)$, $C_{db} = 1/(2\sqrt{1-\gamma/s_b})$, $\kappa_{ij(i>j>1)}=1$ and $\kappa_{ad}, \kappa_{da} \approx 0$, all of which result from the mentioned assumptions.
\begin{gather*}
\langle \rho_{db}^2\rangle = \frac{\eta\gamma^2}{s_d^3}\langle x_b^2x_d^2\rangle = \langle \rho_{dc}^2\rangle = \frac{\eta\gamma^2}{s_d^3}\langle x_c^2x_d^2\rangle = \frac{\eta\gamma^2}{(1-\alpha)^2} \\
\langle \rho_{da}^2\rangle = \frac{\eta\gamma^2(\omega_a+\omega_d)}{2s_d[2-\gamma(\frac{\kappa_{da}}{s_d} + \frac{1}{s_a})]}\frac{(\frac{\kappa_{da}}{s_d} + \frac{1}{s_a})^2}{1+\kappa_{da}^2}\langle x_a^2x_d^2\rangle_x \overset{\substack{\alpha \gg \gamma\\ \gamma \ll 1}}{\approx} \frac{\eta\gamma^2}{2s_a^2s_d}\langle x_a^2x_d^2\rangle_x, \\
\langle \rho_{ad}^2\rangle = \frac{\eta\gamma^2(\omega_a+\omega_d)}{2s_a[2-\gamma(\frac{\kappa_{ad}}{s_a} + \frac{1}{s_d})]}\frac{(\frac{\kappa_{ad}}{s_a} + \frac{1}{s_d})^2}{1+\kappa_{ad}^2}\langle x_a^2x_d^2\rangle_x  \overset{\substack{\alpha \gg \gamma\\ \gamma \ll 1}}{\approx} \frac{\eta\gamma^2\langle x_a^2x_d^2\rangle_x}{2s_as_d^2},
\\ \langle \rho_{ad}\rho_{da}\rangle = \frac{\eta\gamma^2(\omega_a+\omega_d)}{[2s_a-\gamma(\kappa_{ad} + \frac{s_a}{s_d})] +[2s_d-\gamma(\kappa_{da} + \frac{s_d}{s_a})]}\frac{(\frac{\kappa_{ad}}{s_a} + \frac{1}{s_d})(\frac{\kappa_{da}}{s_d} + \frac{1}{s_a})}{\sqrt{1+\kappa_{ad}^2}\sqrt{1+\kappa_{da}^2}}\langle x_a^2x_d^2\rangle_x 
 \overset{\substack{\alpha \gg \gamma\\ \gamma \ll 1}}{\approx} \frac{\eta\gamma^2\langle x_a^2x_d^2\rangle_x}{(s_a+s_d)s_as_d}
\end{gather*}
\renewcommand{\arraystretch}{2}
\begin{center}
\[
\begin{array}{ll}
 \G^{a,b}_{db} = \frac{\gamma}{4}\frac{x_ax_d\sqrt{\omega_a}(\frac{1}{s_a} + \frac{1}{s_d})}{\omega_a+\omega_d} & \Gbar_{db,db}^{a,b} =  \frac{\gamma^2}{16}\frac{\langle x_a^2x_d^2\rangle_x\omega_a(\frac{1}{s_a} + \frac{1}{s_d})^2}{(\omega_a+\omega_d)^2} \overset{\gamma \ll 1}{\approx} \frac{\gamma^2}{64}\frac{(\frac{s_d}{s_a} + 1)^2\langle x_a^2x_d^2\rangle_x}{s_d^2} \nonumber \\ 
 \G^{a,b}_{da} = \frac{-\gamma}{2}\frac{x_bx_d\sqrt{\omega_a\omega_d}}{s_b(\omega_a + \omega_d)^{3/2}}\frac{1 + \kappa_{da}}{\sqrt{1 + \kappa_{da}^2}}  & \Gbar^{a,b}_{da,da} = \frac{\gamma^2}{4}\frac{\langle x_b^2x_d^2\rangle_x \omega_a\omega_d}{s_b^2(\omega_a + \omega_d)^3}\frac{(1 + \kappa_{da})^2}{1 + \kappa_{da}^2} \overset{\substack{\alpha \gg \gamma\\ \gamma \ll 1}}{\approx} \frac{\gamma^2}{32} \frac{\langle x_b^2x_d^2\rangle_x}{s_d^2} \nonumber \\ 
 \G^{a,b}_{ad} = \frac{-\gamma}{2}\frac{x_bx_d\sqrt{\omega_a\omega_d}}{s_b(\omega_a+\omega_d)^{3/2}}\frac{1 + \kappa_{ad}}{\sqrt{1 + \kappa_{ad}^2}} & \Gbar^{a,b}_{ad,ad} = \frac{\gamma^2}{4}\frac{\langle x_b^2x_d^2\rangle_x \omega_a\omega_d}{s_b^2(\omega_a+\omega_d)^3}\frac{(1 + \kappa_{ad})^2}{1 + \kappa_{ad}^2} \overset{\substack{\alpha \gg \gamma\\ \gamma \ll 1}}{\approx} \frac{\gamma^2}{32} \frac{\langle x_b^2x_d^2\rangle_x}{s_d^2} \nonumber \\
 & \Gbar^{a,b}_{ad,da} =\frac{\gamma^2}{4}\frac{\langle x_b^2x_d^2\rangle_x \omega_a\omega_d}{s_b^2(\omega_a+\omega_d)^3}\frac{(1 + \kappa_{da})}{\sqrt{1 + \kappa_{da}^2}}\frac{(1 + \kappa_{ad})}{\sqrt{1 + \kappa_{ad}^2}} \overset{\substack{\alpha \gg \gamma\\ \gamma \ll 1}}{\approx} \frac{\gamma^2}{32} \frac{\langle x_b^2x_d^2\rangle_x}{s_d^2} \nonumber \\
 \G^{c,b}_{db} = \frac{\gamma}{4} \frac{x_cx_d}{s_d\sqrt{\omega_d}}, & \Gbar^{c,b}_{db,db} = \frac{\gamma^2}{16} \frac{\langle x_c^2x_d^2 \rangle_x}{s_d^2\omega_d} = \frac{\gamma^2}{16}\frac{1}{(1-\alpha)\omega_d}\nonumber\\
 \G^{c,b}_{dc}  = \frac{-\gamma }{4} \frac{x_bx_d}{s_d\sqrt{\omega_d}}, &  \Gbar^{c,b}_{dc,dc} =  \frac{\gamma^2 }{16} \frac{\langle x_b^2x_d^2 \rangle_x}{s_d^2\omega_d} = \frac{\gamma^2}{16}\frac{1}{(1-\alpha)\omega_d}
 \nonumber
\end{array}
\]
\end{center}
Eq.\ref{Ds_largealpha} in the main text can be derived from the above by using $D_a \approx n D_{ab}$ and $D_b \approx nD_{bc}$. These results reproduce the numerically measured curves of diffusion as a function of $\alpha$ very well.

\end{document}